\makeatletter
\def\input@path{{template/}{sections/}}
\makeatother

\documentclass[sigconf, nonacm]{acmart}

\newcommand\vldbdoi{XX.XX/XXX.XX}
\newcommand\vldbpages{XXX-XXX}
\newcommand\vldbvolume{14}
\newcommand\vldbissue{1}
\newcommand\vldbyear{2020}
\newcommand\vldbauthors{Jianting Zhang,
    Lefteris Kokoris-Kogias,
    Tasos Kichidis,
    Arun Koshy,
    Mingwei Tian,
    Ilya Sergey,
    Alberto Sonnino}
\newcommand\vldbtitle{\shorttitle}
\newcommand\vldbavailabilityurl{https://anonymous.4open.science/r/beluga-0F4F}
\newcommand\vldbpagestyle{plain}

\newif\ifpublish
\publishtrue

\newif\ifarxiv
\arxivtrue

\usepackage[most]{tcolorbox}
\usepackage{algorithmicx}
\usepackage{algorithm}
\usepackage{amsmath}
\usepackage{amsthm}
\usepackage{array}
\usepackage[noend]{algpseudocode}
\usepackage{balance}
\usepackage{booktabs}
\usepackage{boxedminipage}
\usepackage{cleveref}
\usepackage{color}
\usepackage{csvsimple}
\usepackage{enumitem}
\usepackage{epsfig}
\usepackage{fancyhdr}
\usepackage{float}
\usepackage{graphics}
\usepackage{graphicx}
\usepackage{hyperref}
\usepackage{latexsym}
\usepackage{makecell}
\usepackage{mathtools}
\usepackage{multirow}
\usepackage{nicefrac}
\usepackage{pifont}
\usepackage{pgfplots}
\usepackage{pgfplotstable}
\usepackage{setspace}
\usepackage{siunitx}
\usepackage{soul}
\usepackage{subcaption}
\usepackage{textcomp}
\usepackage{threeparttable}
\usepackage{tikz}
\usepackage{url}
\usepackage{wrapfig}
\usepackage{xcolor}
\usepackage{xparse}
\usepackage{xspace}
\usepackage{listings}
\usepackage{caption}

\newtheorem{assumption}{Assumption}
\newtheorem{theorem}{Theorem}
\newtheorem{lemma}{Lemma}
\newtheorem{corollary}{Corollary}

\newtheorem{definition}{Definition}

\usetikzlibrary{
    shapes.geometric,
    arrows,
    external,
    pgfplots.groupplots,
    matrix
}

\pgfplotsset{compat=1.9}

\algrenewcommand\textproc{}
\algdef{SE}[EVENT]{Event}{EndEvent}[1]{\textbf{upon}\ #1\ \algorithmicdo}{\algorithmicend\ \textbf{event}}%
\algtext*{EndEvent}

\newcommand{\para}[1]{\vspace{0.5em}\noindent\textbf{#1}}

\newcommand{\sysname}{Beluga\xspace} 

\newcommand{\syncopt}{DAG-based uncertified\xspace}
\newcommand{\synccert}{DAG-based certified\xspace}
\newcommand{\syncrbc}{DAG-based RBC\xspace}
\newcommand{\syncmt}{Multi-chain certified\xspace}

\newcommand{\dagbc}{\mathit{block\_propose}}
\newcommand{\dagdeli}{\mathit{block\_store}}
\newcommand{\dagpoa}{\mathit{block\_accept}}

\newcommand{\codelink}{
    \ifpublish
        \url{https://github.com/asonnino/beluga/tree/mysticeti-pull-attack} (commit \texttt{40ebbb8})
    \else
        \url{https://anonymous.4open.science/r/beluga-0F4F}
    \fi
}
\newcommand{\dashboardlink}{
    \ifpublish
        \url{https://github.com/asonnino/beluga/blob/mysticeti-pull-attack/crates/orchestrator/assets/grafana-dashboard.json} (commit \texttt{
            40ebbb8
        })
    \else
        \url{https://anonymous.4open.science/r/beluga-0F4F/crates/orchestrator/assets/grafana-dashboard.json}
    \fi
}

\tcbset{
    highlight/.style={
            enhanced,
            breakable,
            colback=gray!10,
            frame hidden,
            boxrule=0pt,
            borderline west={2pt}{0pt}{black},
            left=2mm, right=2mm, top=1mm, bottom=1mm,
        }
}

\setlist[itemize]{leftmargin=*,itemsep=1pt,topsep=2pt,parsep=0pt}
\setlist[enumerate]{leftmargin=*,itemsep=1pt,topsep=2pt,parsep=0pt}

\definecolor{eclipseStrings}{RGB}{42,0.0,255}
\definecolor{eclipseKeywords}{RGB}{127,0,85}
\colorlet{numb}{magenta!60!black}
\lstdefinelanguage{json}{
    basicstyle=\normalfont\ttfamily,
    commentstyle=\color{eclipseStrings},
    stringstyle=\color{eclipseKeywords},
    stepnumber=1,
    numbersep=8pt,
    showstringspaces=false,
    breaklines=true,
    string=[s]{"}{"},
    comment=[l]{:\ "},
    morecomment=[l]{:"},
    literate=
        *{0}{{{\color{numb}0}}}{1}
        {1}{{{\color{numb}1}}}{1}
        {2}{{{\color{numb}2}}}{1}
        {3}{{{\color{numb}3}}}{1}
        {4}{{{\color{numb}4}}}{1}
        {5}{{{\color{numb}5}}}{1}
        {6}{{{\color{numb}6}}}{1}
        {7}{{{\color{numb}7}}}{1}
        {8}{{{\color{numb}8}}}{1}
        {9}{{{\color{numb}9}}}{1}
}

\newcommand\YAMLcolonstyle{\color{red}\mdseries}
\newcommand\YAMLkeystyle{\color{black}\bfseries}
\newcommand\YAMLvaluestyle{\color{blue}\mdseries}
\makeatletter
\newcommand\language@yaml{yaml}
\expandafter\expandafter\expandafter\lstdefinelanguage
\expandafter{\language@yaml}
{
    keywords={true,false,null,y,n},
    keywordstyle=\color{darkgray}\bfseries,
    basicstyle=\YAMLkeystyle,
    sensitive=false,
    comment=[l]{\#},
    morecomment=[s]{/*}{*/},
    commentstyle=\color{purple}\ttfamily,
    stringstyle=\YAMLvaluestyle\ttfamily,
    moredelim=[l][\color{orange}]{\&},
    moredelim=[l][\color{magenta}]{*},
    moredelim=**[il][\YAMLcolonstyle{:}\YAMLvaluestyle]{:},
    morestring=[b]',
    morestring=[b]",
    literate =    {---}{{\ProcessThreeDashes}}3
        {>}{{\textcolor{red}\textgreater}}1
        {|}{{\textcolor{red}\textbar}}1
        {\ -\ }{{\mdseries\ -\ }}3,
}
\lst@AddToHook{EveryLine}{\ifx\lst@language\language@yaml\YAMLkeystyle\fi}
\makeatother
\newcommand\ProcessThreeDashes{\llap{\color{cyan}\mdseries-{-}-}}

\begin{document}

\title{\sysname: Block Synchronization for BFT Consensus Protocols}

\author{Jianting Zhang$^{4}$,
    Lefteris Kokoris-Kogias$^{1}$,
    Tasos Kichidis$^{1}$,
    Arun Koshy$^{1}$,
    Mingwei Tian$^{1}$,
    Ilya Sergey$^{1,3}$,
    Alberto Sonnino$^{1,2}$}

\affiliation{%
    \institution{$^{1}$Mysten Labs;
        $^{2}$University College London;
        $^{3}$National University of Singapore;
        $^{4}$Purdue University}
    \country{}
}
\email{}

\begin{abstract}

    Modern high-throughput BFT consensus protocols use streamlined push-pull mechanisms to disseminate blocks and keep happy-path performance optimal. Yet state-of-the-art designs lack a principled and efficient way to exchange blocks, which leaves them open to targeted attacks and performance collapse under network asynchrony.
    This work introduces the \emph{block synchronizer}, a simple abstraction that drives incremental block retrieval and enforces resource-aware exchange. Its interface and role sit cleanly inside a modern BFT consensus stack. We also uncover a new attack, where an adversary steers honest validators into redundant, uncoordinated pulls that exhaust bandwidth and stall progress.
    \sysname is a modular and scarcity-aware instantiation of the block synchronizer. It achieves optimal common-case latency while bounding the cost of recovery under faults and adversarial behavior. We integrate \sysname into Mysticeti, the consensus core of the Sui blockchain, and show on a geo-distributed AWS deployment that \sysname sustains optimal performance in the optimistic path and, under attack, delivers up to $3\times$ higher throughput and $25\times$ lower latency than prior designs. The Sui blockchain adopted \sysname in production.

\end{abstract}

\maketitle

\pagestyle{\vldbpagestyle}
\begingroup\small\noindent\raggedright\textbf{PVLDB Reference Format:}\\
\vldbauthors. \vldbtitle. PVLDB, \vldbvolume(\vldbissue): \vldbpages, \vldbyear.\\
\href{https://doi.org/\vldbdoi}{doi:\vldbdoi}
\endgroup
\begingroup
\renewcommand\thefootnote{}\footnote{\noindent
    This work is licensed under the Creative Commons BY-NC-ND 4.0 International License. Visit \url{https://creativecommons.org/licenses/by-nc-nd/4.0/} to view a copy of this license. For any use beyond those covered by this license, obtain permission by emailing \href{mailto:info@vldb.org}{info@vldb.org}. Copyright is held by the owner/author(s). Publication rights licensed to the VLDB Endowment. \\
    \raggedright Proceedings of the VLDB Endowment, Vol. \vldbvolume, No. \vldbissue\ %
    ISSN 2150-8097. \\
    \href{https://doi.org/\vldbdoi}{doi:\vldbdoi} \\
}\addtocounter{footnote}{-1}\endgroup
\ifdefempty{\vldbavailabilityurl}{}{
    \vspace{.3cm}
    \begingroup\small\noindent\raggedright\textbf{PVLDB Artifact Availability:}\\
    The source code, data, and/or other artifacts have been made available at \url{\vldbavailabilityurl}.
    \endgroup
}

\section{Introduction}\label{sec:introduction}

The last decade of research in high-throughput Byzantine fault-tolerant (BFT) consensus protocols has unveiled that achieving world-class performance requires two key design choices.
Firstly, modern blockchains decouple block ordering from bulk data dissemination~\cite{yang2019prism,danezis2022narwhal}, and secondly, they chain the disseminated data with causal references to past disseminated data~\cite{buchman2016tendermint,cohen2022aware,danezis2022narwhal}. Both properties depart from the blueprint used by legacy protocols such as PBFT~\cite{castro1999practical}, where data is disseminated solely by a rotating leader using monotonically increasing view numbers; instead, they adopt a concurrent data dissemination approach. This means each validator is expected to assemble transactions into blocks or batches and broadcast them to all other validators~\cite{danezis2022narwhal}.

This dissemination usually imposes chaining or causal dependencies between blocks, but protocols implement this with varying degrees of rigidity.
Some protocols, such as Narwhal-Hotstuff~\cite{danezis2022narwhal} and Autobahn~\cite{giridharan2024autobahn} achieve this through
direct dissemination, 
after which a separate module imposes a causal order over 
this data.
Others, such as HashGraph~\cite{green2021hashgraph} and Blocklace-based systems~\cite{blocklace,grass-route} have 
validators assemble transactions into a block 
and disseminate the block by referencing as many blocks 
as possible.
Finally, protocols such as Bullshark~\cite{spiegelman2022bullshark} and Mysticeti~\cite{babel2025mysticeti} force validators  
to only disseminate blocks upon collecting at least a quorum of blocks from the previous round.
Regardless of the approach, to implement atomic broadcast~\cite{cachin2011introduction}, 
they must implement the reliable broadcast abstraction.
Specifically, even if the BFT consensus protocol ensures consistency, the dissemination layer must still enforce a liveness-relevant property---totality~\cite{cachin2011introduction}: if one honest validator receives a block, then all other validators must be able to obtain it as well.


Through manual inspection of numerous high-performant BFT codebases~\cite{sui-code,narwhal-code,hotstuff-code,mysticeti-code,jolteon-code,diem-code,autobahn-code,mahi-mahi-code,ditto-code}, we observed that contrary to the liveness proofs and descriptions of these protocols, none of the state-of-the-art protocols actually implement an upfront reliable broadcast, such as Bracha broadcast~\cite{bracha1987asynchronous}. That is, they do not implement the totality property through a double-echo. The reason is clear: this would be prohibitively expensive as it consumes precious bandwidth, and the double-echo is rarely needed (only upon faults or poor network conditions)~\cite{danezis2022narwhal,shrestha2025towards}.
We empirically uncover the \emph{same} implicit two-phase pattern: an optimistic \emph{push} of block identifiers 
followed by a probabilistic \emph{pull} for missing blocks. We confirm this through a systematic inspection of diverse production and prototype codebases of state-of-the-art protocols, examining dissemination modules, recovery loops, timers, and traces. The optimistic push is implemented through weaker primitives such as Byzantine consistent broadcast (lacking \emph{totality})~\cite{Bullshark-code,narwhal-code,sailfish-code} or even best-effort broadcast~\cite{mysticeti-code,sui-code}. A recovery mechanism asynchronously \emph{pulls} any missing blocks deemed ``useful''. 


\para{The missing component.}
Despite its importance, this push-pull mechanism remains unstructured glue code, often involving uncoordinated pre-dissemination of block identifiers, unoptimized random pulls, and no explicit bounds on recovery message complexity (\autoref{sec-basic-sync}).
Protocols run ad hoc logic on received blocks to determine whether and when to trigger the pull part of the protocol.
This hidden component governs throughput under load, tail latency, and resilience against adversarial behavior and network conditions. Yet it lacks a specification or provable bounds.

As an almost expected consequence of this discrepancy between theory and practice, we uncover how an adversary can induce a targeted performance degradation that we call the \emph{pull induction attack}. Byzantine replicas selectively withhold messages during initial dissemination so that only a small subset of honest validators receive a block. This then triggers redundant and overlapping pull requests from the remaining honest validators. The attack can be repeated every round and can drastically increase latency and consume recovery bandwidth (\autoref{sec-attack-insights}).
We even observe that some implementations entirely omit the pull mechanism, thus claiming happy path performance while silently failing to achieve liveness under adversarial settings.


In this work, we formalize the module of modern high-throughput BFT consensus responsible for block synchronization, which we call the \emph{block synchronizer} (\autoref{sec-block-sync-def}). 
This module decides how validators push their blocks to other validators, how to fetch any missing causal dependencies, and runs a principled admission control on whether to include the block in its local dependency graph (i.e., use it as a parent).
An ideal block synchronizer module must satisfy two main goals: (G1) optimal push latency along the optimistic path; and (G2) bounded amplification even under adversarial scheduling.
Existing block synchronizer protocols fall short on G1 due to conservative multi-step push procedures (such as explicit consistent broadcast) and/or on G2 because of ad hoc push-pull designs. We provide a summary of these protocols in \autoref{sec-basic-sync} (\Cref{tab:comparison}).

We then present \sysname, a block synchronizer module that satisfies all the goals above (\autoref{sec-sysname}). \sysname is modular and integrates into existing protocols without altering their safety guarantees or ordering logic.
%
%
The key insight behind the construction of \sysname is that heavyweight upfront dissemination proofs (such as running an explicit Bracha broadcast) are \emph{sufficient but not necessary} for robustness. Instead, it structures the push phase to carefully select the causal history of blocks based on scarcity signals (as opposed to simply selecting the first that arrives) and opportunistically leverages implicit dissemination evidence. The detection of missing blocks is performed through the reception of messages containing digests of past blocks unknown to the validator. In other words, the pull mechanism leverages the chaining and causal history of messages. Together, these mechanisms enable scarcity-aware pulls that provably bound recovery complexity while retaining optimal optimistic-path performance. In short, \sysname aims to preserve the optimistic high performance of current implementations while reintroducing the robustness guarantees of classic reliable broadcast.


\para{Real world impact.}
We implement \sysname within Mysticeti, the consensus core of the Sui blockchain~\cite{sui}. Our evaluation on a geo-distributed AWS testbed shows that \sysname maintains optimal optimistic-path latency and throughput while significantly improving performance under targeted attacks and adverse network conditions. Under attack scenarios, \sysname improves throughput by over 3$\times$ and reduces latency by 25$\times$ compared to baseline Mysticeti with its default block synchronizer (\autoref{sec-experiment}).
We collaborated with the Sui team for over a year to integrate \sysname into the Sui blockchain. \sysname has been live since January 2025 with \texttt{mainnet-v1.42.0}, helping to secure over 5B USD in assets. \Cref{fig:production-latency} shows the direct impact of \sysname: under network degradation attacks, Sui's tail latency improved by 5x after deploying \sysname.

\begin{figure}[t]
    \centering
    \includegraphics[height=4.9cm]{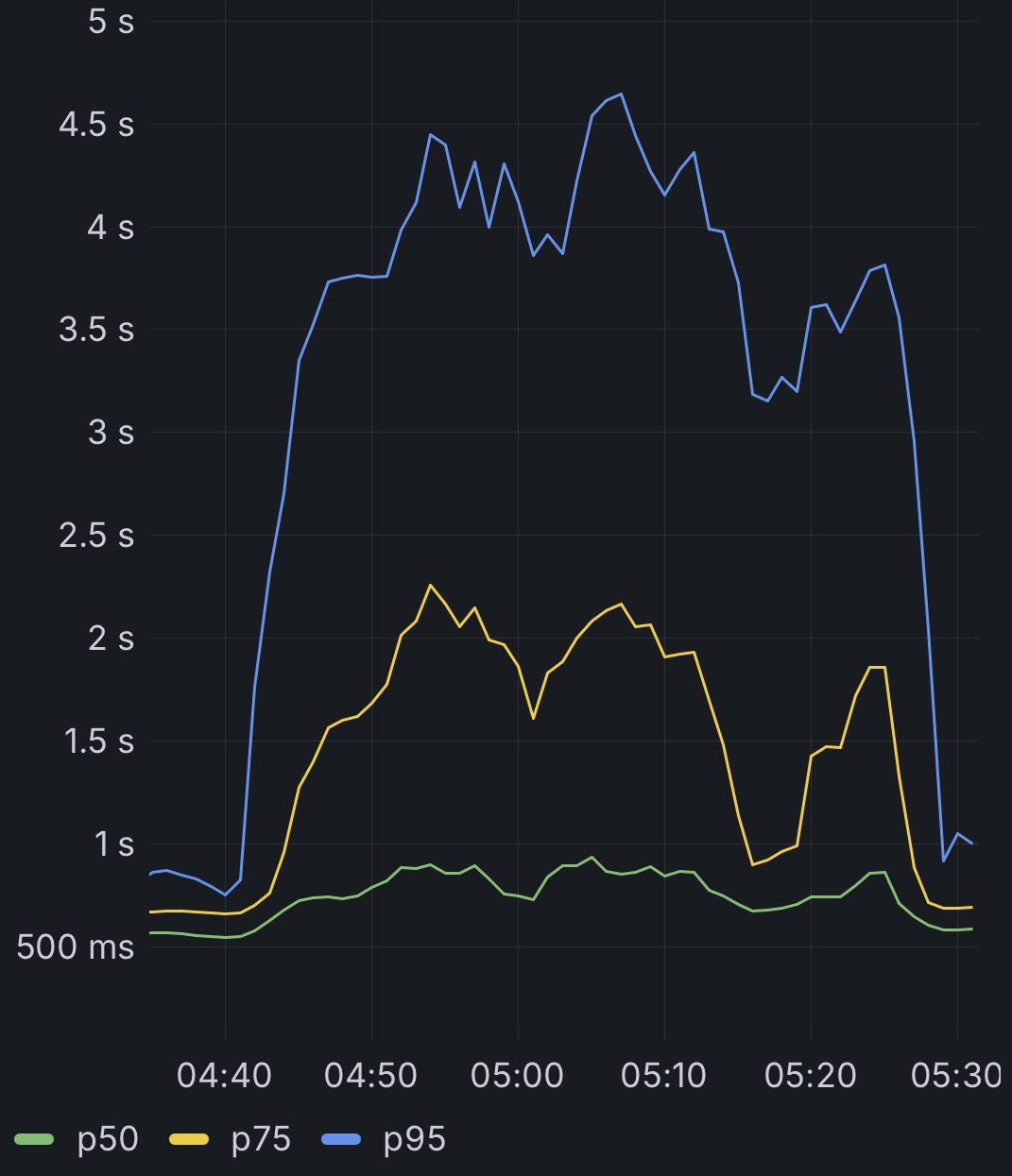}
    \includegraphics[height=4.9cm]{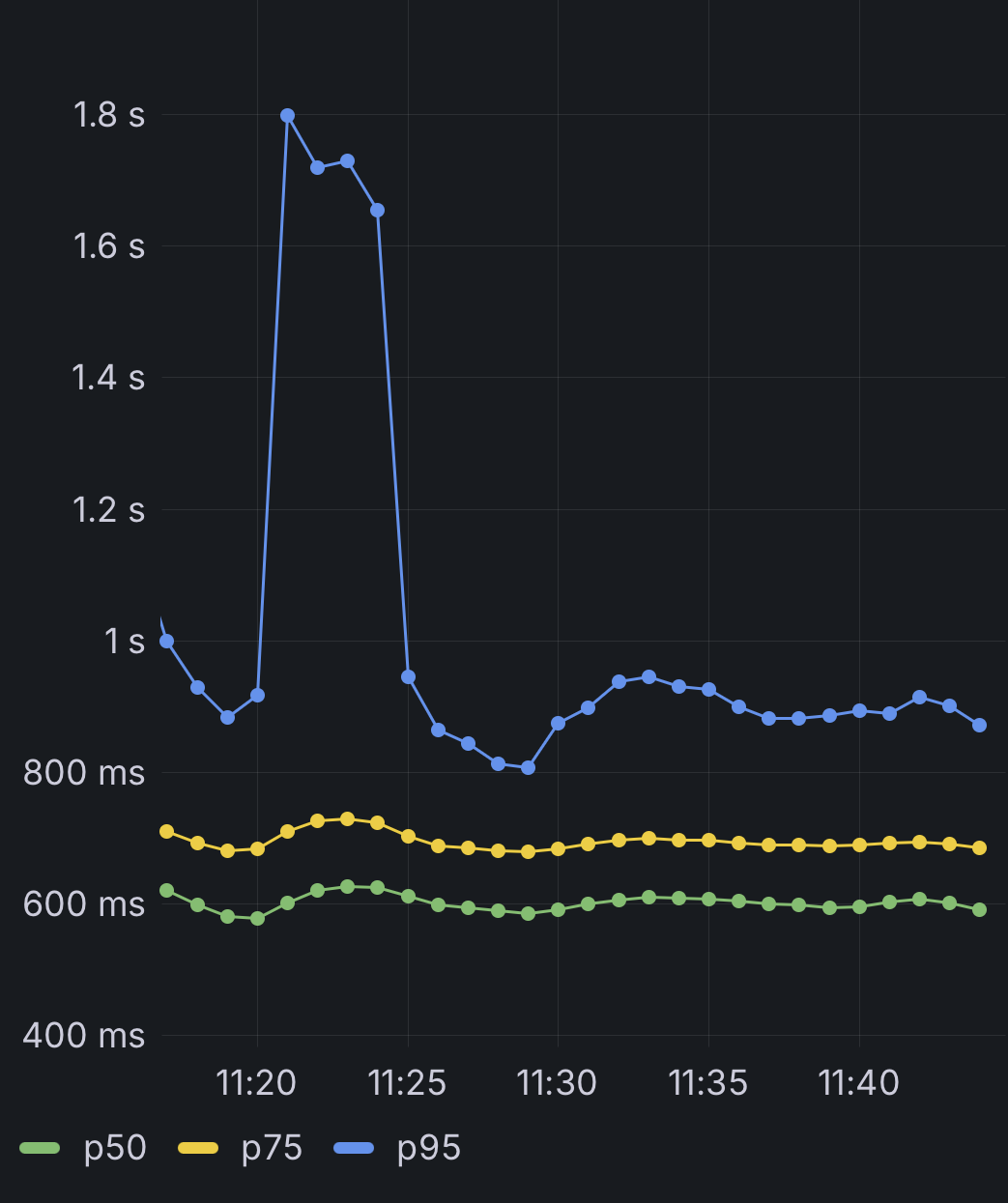}
    \caption{
        Latency of the Sui blockchain under network attacks before (left) and after (right) deploying \sysname.
    }
    \label{fig:production-latency}
\end{figure}


\para{Contributions.} We make the following contributions:
\begin{enumerate}
    \item We formalize the block synchronizer interface and threat model, and derive baseline limitations of existing ad hoc mechanisms.
    \item We design and analyze the \emph{pull induction attack}, which increases latency by up to 50$\times$ and degrades throughput by up to 15$\times$ in existing systems.
    \item We propose \sysname, a structured synchronizer with scarcity-aware detection and diversity-maximizing pull scheduling that preserves optimistic latency while bounding amplification.
    \item We provide a formalisation of \sysname's main properties in Lean prover~\cite{Moura021} with machine-checked proofs of all non-probabilistic theorems in this paper.
    \item We implement and integrate \sysname into Mysticeti/Sui, demonstrating no optimistic-path regression and significant under-attack throughput (3$\times$)  and latency (25$\times$) improvements.
\end{enumerate}
\section{Problem Definition}\label{sec-problem-def}
\para{Network model:}
We assume a set $\mathcal{V}$ of $n=3f+1$ validators (or parties; both are used interchangeably throughout this work) $\{v_0,\dots, v_{n-1}\}$ and a static adversary $\mathcal{A}$ that can corrupt up to $f$ of the parties arbitrarily, at any point.
A party is \emph{crashed} if it halts prematurely at some point during execution. Parties that deviate arbitrarily from the protocol are called \emph{Byzantine} or bad.
Parties that are never crashed or Byzantine are called \emph{honest}.
Parties are communicating over a partially synchronous network~\cite{JACM:DwoLynSto}, in which there exists a special event called Global Stabilization Time ($GST$) and a known finite time bound $\Delta$ (we use $\delta$ to represent the actual network latency), such that any message sent by a party at time $x$ is guaranteed to arrive by time $\Delta+ \max\{GST, x\}$.

\para{Threat model:}
The adversary is computationally bounded. Pairwise points of communication between any two honest parties are considered \emph{reliable}, i.e., any honest message is \emph{eventually} (after a finite, bounded number of steps) delivered.
However, until $GST$ the adversary controls the delivery of all messages in the network, with the only limitation that the messages must be eventually delivered. After $GST$, the network becomes synchronous, and messages are guaranteed to be delivered within $\Delta$ time after the time they are sent, potentially in an adversarially chosen order.

\subsection{The Block Synchronizer Problem} \label{sec-block-sync-def}
We formalize the block synchronization problem here. We call it \emph{block synchronizer}. As preliminaries, a block $B$ contains the following basic data fields: round number $r$, block digest $d$, block creator $author$, connected parent blocks (in digest) $parents$, a list of transactions $payload$, and the creator's $signature$ on $B$.
\begin{definition}[Block synchronizer] \label{def-dag-sync}
    In the block synchronizer problem, a group of $n$ validators $\mathcal{V}$ (of which up to $f$ are Byzantine) collectively builds a structured, non-empty, and ever-growing set of blocks that are indexed by a monotonically increasing round number. Each validator $v_i$ can call $\dagbc_i(B, r)$ to push its block $B$ in round $r$ (where $B.r{=}r$) to the system. Each validator $v_i$ outputs $\dagpoa_i(B.d)$ to accept $B$. The block $B$ must include $p$ blocks (from distinct validators) in round $B.r{-}1$ as $B.parents$ that have been output via $\dagpoa_i$ by $v_i$, where $p$ is a parameter specifying how blocks are structured. Each validator $v_i$ then outputs $\dagdeli_i(B)$ to store $B$. The protocol should satisfy the following properties:
\end{definition}


\begin{itemize}


    \item \emph{\textbf{Round-Progression:}}
          In each global round $r$ of the system, at least $2f{+}1$ validators (not necessarily honest) invoke $\dagbc$ to create and disseminate their blocks. In other words, for each given $r {\geq} 0$, at least $2f{+}1$ produce blocks that have $r$ as their round.

    \item \emph{\textbf{Round-Termination:}}
          For each global round $r$ of the system, every honest validator eventually outputs $\dagpoa$ for the blocks produced in this round by least $2f{+}1$ validators. That is, for each given $r \geq 0$, each honest validator accepts block proposals, whose assigned round is $r$, from at least $2f{+}1$ different validators.
    \item \textbf{\emph{Block availability}}:
          For some block $B$ produced in round $r$, if an honest validator $v_i$ outputs $\dagpoa_i(B.d)$, then $v_i$ eventually outputs $\dagdeli_i(B)$.
    \item \emph{\textbf{Causal availability}}:
          For some block $B$ produced in round $r$, if an honest validator $v_i$ outputs $\dagpoa_i(B.d)$, then for every block $B' \in causal(B)$, $v_i$ eventually outputs $\dagpoa_i(B'.d)$, where $causal(B)$ represents $B$'s causal history (i.e., all blocks for which there is a connection or path from $B$ to them).
\end{itemize}

\begin{table*}[!tbp]
    \centering
    \caption{Comparison of synchronizer protocols and their integrations to BFT consensus, after GST and the leader is honest}
    \footnotesize
    \label{tab:comparison}
    \begin{threeparttable}
        \begin{tabular}{@{}lccccccccc@{}}
            \toprule
            \multirow{2}{*}[-5pt]{\makecell{Consensus                                                                                                                            \\ Protocols}} & \multirow{2}{*}[-5pt]{\makecell{Synchronizer \\ Protocols}} & \multicolumn{3}{c}{Optimistic Case}                                                                   & \multicolumn{3}{c}{Adverse Case\tnote{(2)}} & \multirow{2}{*}[-5pt]{\makecell{Optimal \\ Push?}} & \multirow{2}{*}[-5pt]{\makecell{Bounded \\ Amplification?}}                                                                                           \\
            \cmidrule(lr){3-5} \cmidrule(lr){6-8}
                                                                                                                                                                                                &                                                             & \makecell{Push                                                                            \\ Latency} & \makecell{Consensus \\ Latency\tnote{(1)}}  & \makecell{Pull \\ Complexity}                      & \makecell{Push \\ Latency}                                  & \makecell{Consensus \\ Latency} & \makecell{Pull \\ Complexity}                         \\
            \midrule
            \multirow{2}{*}[0pt]{Star~\cite{duan2024dashing}}                                                                                                                                   & \syncmt                                                     & 2$\delta$                                                                                             & 5$\delta$                                   & $O(n)$                                             & 2$\delta$                                                   & 5$\delta$                       & $O(n)$                        & \ding{55} & \ding{51} \\
            \cmidrule(lr){2-10}
                                                                                                                                                                                                & \sysname                                                    & $\delta$                                                                                              & 5$\delta$                                   & $O(1)$                                             & $2\kappa\delta$\tnote{(3)}                                  & $4\kappa\delta{+}3\delta$       & $O(n)$                        & \ding{51} & \ding{51} \\
            \midrule
            \multirow{2}{*}[0pt]{Autobahn~\cite{giridharan2024autobahn}}                                                                                                                        & \syncmt                                                     & 2$\delta$                                                                                             & 7$\delta$                                   & $O(n)$                                             & 2$\delta$                                                   & 7$\delta$                       & $O(n)$                        & \ding{55} & \ding{51} \\
            \cmidrule(lr){2-10}
                                                                                                                                                                                                & \sysname                                                    & $\delta$                                                                                              & 7$\delta$                                   & $O(1)$                                             & $2\kappa\delta$                                             & $4\kappa\delta{+}5\delta$       & $O(n)$                        & \ding{51} & \ding{51} \\
            \midrule
            \multirow{2}*{Bullshark~\cite{spiegelman2022bullshark}}                                                                                                                             & \synccert                                                   & 3$\delta$                                                                                             & 6$\delta$                                   & $O(n)$                                             & 5$\delta$                                                   & 10$\delta$                      & $O(n)$                        & \ding{55} & \ding{51} \\
            \cmidrule(lr){2-10}
                                                                                                                                                                                                & \sysname                                                    & $\delta$                                                                                              & 4$\delta$                                   & $O(1)$                                             & $2\kappa\delta$                                             & $8\kappa\delta$                 & $O(n)$                        & \ding{51} & \ding{51} \\
            \midrule
            \multirow{2}*{Shoal++~\cite{arun2024shoal++}}                                                                                                                                       & \synccert                                                   & 3$\delta$                                                                                             & 4$\delta$                                   & $O(n)$                                             & 5$\delta$                                                   & 8$\delta$                       & $O(n)$                        & \ding{55} & \ding{51} \\
            \cmidrule(lr){2-10}
                                                                                                                                                                                                & \sysname                                                    & $\delta$                                                                                              & 3$\delta$                                   & $O(1)$                                             & $2\kappa\delta$                                             & $6\kappa\delta$                 & $O(n)$                        & \ding{51} & \ding{51} \\
            \midrule
            \multirow{3}*{Sailfish~\cite{shrestha2024sailfish}}                                                                                                                                 & \syncrbc~\cite{abraham2021good}                                                   & 2$\delta$                                                                                             & 3$\delta$                                   & None                                               & 2$\delta$                                                   & 3$\delta$                       & None                          & \ding{55} & \ding{55} \\
                            & \syncrbc~\cite{das2021asynchronous}                                                    & 4$\delta$                                                                                             & 5$\delta$                                   & None                                               & 4$\delta$                                                   & 5$\delta$                       & None                          & \ding{55} & \ding{55} \\
            \cmidrule(lr){2-10}
                                                                                                                                                                                                & \sysname                                                    & $\delta$                                                                                              & 3$\delta$                                   & $O(1)$                                             & $2\kappa\delta$                                             & $6\kappa\delta$                 & $O(n)$                        & \ding{51} & \ding{51} \\
            \midrule
            \multirow{2}*{Sailfish++~\cite{shrestha2025optimistic}}                                                                                                                             & \syncrbc                                                    & 2$\delta$                                                                                             & 3$\delta$                                   & None                                               & 3$\delta$                                                   & 4$\delta$                       & None                          & \ding{55} & \ding{55} \\
            \cmidrule(lr){2-10}
                                                                                                                                                                                                & \sysname                                                    & $\delta$                                                                                              & 3$\delta$                                   & $O(1)$                                             & $2\kappa\delta$                                             & $6\kappa\delta$                 & $O(n)$                        & \ding{51} & \ding{51} \\
            \midrule
            \multirow{2}*{Mysticeti~\cite{babel2025mysticeti}}                                                                                                                                  & \syncopt                                                    & $\delta$                                                                                              & 3$\delta$                                   & $O(1)$                                             & $\geq$$3\delta$          & $\geq$$9\delta$                  & $O(1)$                          & \ding{51}                     & \ding{55}             \\
            \cmidrule(lr){2-10}
                                                                                                                                                                                                & \sysname                                                    & $\delta$                                                                                              & 3$\delta$                                   & $O(1)$                                             & $2\kappa\delta$                                             & $6\kappa\delta$                 & $O(n)$                        & \ding{51} & \ding{51} \\
            \bottomrule
        \end{tabular}
        \begin{tablenotes}
            \footnotesize
            \item[{(1)}] Consensus latency specifies how long the protocol requires to reach a consensus on transactions of a leader block.
            \item[{(2)}] The adverse case is when the adversary can \emph{successfully} stall protocol progress via pull-induction attacks (detailed in \autoref{sec-attack}) or by deviating from optimistic termination conditions in~\cite{shrestha2024sailfish}. The latency of \sysname under adverse cases is upper bound in theory and might not be reached in practice without assuming a powerful adversary who can arbitrarily reschedule message delivery. More importantly, when honest reputations dominate, \sysname achieves optimal push latency~\cite[Lemma 3]{kichidis2025beluga}
            \item[{(3)}] $\kappa=1+\frac{f}{R_L}$, where $f$ is the number of malicious validators, and $R_L$ is a system parameter (see \autoref{sec-sysname-push} for more details). By setting $R_L {\gg} f$, \sysname achieves $\kappa{=}1$.
        \end{tablenotes}

    \end{threeparttable}
\end{table*}

{
    After comprehensively reviewing the existing relevant high-throughput consensus protocols, we observe that they have some integrated block synchronization modules, each of which is composed of a push protocol and an optional pull protocol. For the sake of comparison, we define several performance-related metrics: (i) \emph{push latency}, which is the network latency required to complete one round of block production, i.e., the interval between the creation of blocks in two consecutive rounds, (ii) \emph{pull latency}, which is the network latency to fetch a block, and (iii) \emph{pull complexity}, which is the communication complexity per node of fetching a block. Intuitively, push latency measures the protocol's data dissemination speed, and we consider push latency to be \emph{optimal} if it equals $\delta$ throughout this paper, while pull latency and complexity indicate the overhead of fetching blocks.
}

{
    With the defined metrics, we consider that an ideal block synchronizer module/protocol must satisfy the following goals:
    \begin{enumerate}[label=\textbf{(G\arabic*)}]
        \item \textbf{Optimal push:} The push latency is optimal under optimistic cases where the protocol progress is not stalled by the adversary and the network is synchronous.
        \item \textbf{Bounded amplification:} Under adverse cases, block retransmission and pull latency caused by the adversary are bounded when the network is synchronous. That is, the adversary cannot trigger infinite block retransmission and unbounded latency to synchronize blocks that have been received by honest validators.
    \end{enumerate}
}

\subsection{Existing Synchronizer Protocols}\label{sec-basic-sync}
In this section, we explore the block synchroniser protocols in the existing works, where we model each class to the best of our abilities by both studying their research and implementations. The comparison is summarized in \Cref{tab:comparison}.

\subsubsection{Multi-chain Certified Synchronizer Protocol}
Many multi-chain BFT protocols, such as Autobahn~\cite{giridharan2024autobahn} and Star~\cite{duan2024dashing}, are built on the \syncmt synchronizer protocol. 
Validators structure blocks into multiple parallel chains, where each block includes only one block from the same validator in the last round as parents (i.e., $p{=}1$). The \syncmt synchronizer protocol consists of a weak quorum-based push protocol to push blocks and a deterministic pull protocol to fetch missing blocks.

\para{Weak quorum-based push.}
For each round $r$, every validator $v_i$ calls $\dagbc_i(B, r)$ to push a round $r$ block $B$ into the system. Specifically, upon outputting $\dagpoa$ for its own block in round $r{-}1$, $v_i$ moves to round $r$ and create a new block $B$ with $v_i$'s round $r{-}1$ block in $B.parents$. Then, $v_i$ employs a two-step Propose-Vote scheme to disseminate $B$.
In the first step, $v_i$ broadcasts $B$. In the second step, other validators respond to $v_i$ with signatures on $B$ and cache $B$. 
After that, $v_i$ uses $f{+}1$ signatures to construct a \emph{weak} certificate $WC(B.d)$ for $B$ and outputs $\dagpoa_i(B.d)$ and $\dagdeli_i(B)$. $WC(B.d)$ will be piggybacked into $v_i$'s round $r{+}1$ block $B'$. The other validator $v_j$ receiving $B'$ outputs $\dagpoa_j(B.d)$ and outputs $\dagdeli_j(B)$ if receiving $B$.


\para{Deterministic pull.}
Validators in the \syncmt synchronizer protocol perform the pull protocol independently from the push protocol. Specifically, when a validator $v_i$ needs to pull a missing block $B$, $v_i$ deterministically chooses a set of $f{+}1$ validators $\mathcal{V}_B$ in $WC(B.d)$ (i.e., validators who sign $B$) and sends a request message to $\mathcal{V}_B$. Since validators in $\mathcal{V}_B$ cache $B$, at least one honest validator can serve as the provider of $B$.
Therefore, $v_i$ must be able to receive $B$. $v_i$ outputs $\dagdeli_i(B)$ once $\dagpoa_i(B.d)$ succeeds and it receives $B$.

The \syncmt synchronizer protocol achieves bounded amplification (G2), but not optimal push latency (G1). Specifically, it achieves a bounded pull latency of $2\delta$ with the deterministic pull protocol and a bounded data retransmission, since each validator stops retransmitting blocks after receiving $f{+}1$ acknowledgments. The pull complexity is $O(n)$. However, the Propose-Vote scheme requires a push latency of $2\delta$, which is not optimal.


\subsubsection{DAG-based Certified Synchronizer Protocol}\label{sec-basic-sync-cert}
The certified DAG protocols, such as Bullshark~\cite{spiegelman2022bullshark} and Shoal++~\cite{arun2024shoal++}, are built on a \synccert synchronizer protocol. 
Validators organize blocks into a directed acyclic graph (DAG) using the quorum-based broadcast primitive. It consists of a consistent broadcast (CBC)-based push protocol to push certified blocks and a deterministic pull protocol to fetch missing blocks.

\para{CBC-based push.}
For each round $r$, every validator $v_i$ pushes a round $r$ block $B$ into the system by calling $\dagbc_i(B, r)$.
Specifically, upon outputting $\dagpoa$ for blocks in round $r{-}1$ from at least $2f{+}1$ validators, $v_i$ moves to round $r$ and creates a new block $B$ with \emph{all} these round $r{-}1$ blocks as $B.parents$ (thus, $p{=}2f{+}1$).
Then, $v_i$ employs a three-step certificate scheme to disseminate $B$. In the first step, $v_i$ broadcasts $B$. In the second step, other validators respond to $v_i$ with signatures on $B$ and cache $B$. 
Then, $v_i$ uses these $2f{+}1$ signatures to construct a certificate $C(B.d)$. In the third step, $v_i$ broadcasts $C(B.d)$. Moreover, when receiving $C(B'.d)$, $v_i$ outputs $\dagpoa_i(B'.d)$ and output $\dagdeli_i(B')$ if receving $B'$.

\para{Deterministic pull.}
The \synccert synchronizer protocol employs a deterministic pull protocol as well.
$v_i$ requests a missing block from the set of validators $\mathcal{V}_B$ in $C(B.d)$. $v_i$ outputs $\dagdeli_i(B)$ after $\dagpoa_i(B.d)$ succeeds and receiving $B$.

The \synccert protocol ensures bounded amplification (G2), but not optimal push latency (G1). Specifically, it achieves a pull latency of $2\delta$ and a bounded data retransmission like the \syncmt protocol. The pull complexity is $O(n)$.
However, the CBC-based push protocol requires a push latency of $3\delta$. 

\subsubsection{\syncrbc Synchronizer Protocol} \label{sec-basic-sync-rbc}
Many recent DAG-based BFT consensus protocols, such as Sailfish~\cite{shrestha2024sailfish} and its variant Sailfish++~\cite{shrestha2025optimistic}, are built on a \syncrbc synchronizer protocol. 
Validators organize blocks into a DAG using the reliable broadcast (RBC) primitive. The RBC ensures that if an honest validator receives a block, all honest validators will eventually receive it (i.e., ensuring the totality property). Therefore, the pull protocol is \emph{theoretically} not needed in the \syncrbc synchronizer protocol, and it only contains an RBC-based push protocol.

\para{RBC-based push.}
For each round $r$, every validator $v_i$ pushes a round $r$ block $B$ into the system by calling $\dagbc_i(B, r)$.
Specifically, upon outputting $\dagpoa$ for blocks in round $r{-}1$ from at least $2f{+}1$ validators, a validator $v_i$ moves to round $r$ and creates a new block $B$ referencing \emph{all} these round $r{-}1$ blocks as $B.parents$. $v_i$ disseminates $B$ using an RBC protocol~\cite{shrestha2025optimistic, bracha1987asynchronous,abraham2021good,das2021asynchronous}. If $v_i$ is honest, according to the RBC's Validity property, every other honest validator $v_j$ will eventually deliver $B$ (i.e., $v_j$ will receive $B$).
In addition, when $v_i$ delivers $B'$ that is reliably broadcast by another validator $v_j$, $v_i$ checks if it can accept $B'$ by checking if it has output $\dagpoa_i(B''.d)$ for every block $B'' \in casual(B')$. If yes, $v_i$ outputs $\dagpoa_i(B'.d)$ and $\dagdeli_i(B')$ right after. If no, $v_i$ will put $B'$ into a pending list and update the list whenever it outputs $\dagpoa$ for a new block.

The \syncrbc synchronizer protocol fails to guarantee both optimal push latency (G1) and bounded amplification (G2). Specifically, implementing an upfront reliable broadcast requires honest validators to continually transmit messages to an unresponsive adversary, leading to unbounded retransmissions. In addition, the push latency depends on the RBC protocol used, but none of them are optimal (i.e., $\delta$). For instance, Sailfish~\cite{shrestha2024sailfish} employs the RBC protocol from~\cite{das2021asynchronous}, leading to a push latency of $4\delta$. Sailfish++~\cite{shrestha2025optimistic} adopts picky RBC protocol to achieve a push latency of $2\delta$ (still not optimal) in optimistic cases (where at least $\lceil\frac{n+2f-2}{2}\rceil$ validators behave honestly). Sailfish++'s RBC protocol has a round latency of $3\delta$ when optimistic cases do not hold.

\subsubsection{DAG-based Uncertified Synchronizer Protocol}\label{sec-basic-sync-opt}
Mysticeti~\cite{babel2025mysticeti} consensus protocol uses a \syncopt synchronizer protocol.
Validators structure blocks into a DAG.
The \syncopt synchronizer protocol consists of a best-effort broadcast (BEB)-based push protocol to push blocks and a random pull protocol to fetch missing blocks.

\para{BEB-based push.}
For each round $r$, every validator $v_i$ pushes a round $r$ block $B$ into the system by calling $\dagbc_i(B, r)$.
When receiving a block $B'$ from other validators, $v_i$ checks if it has output $\dagpoa_i(B''.d)$ for every block $B'' \in casual(B')$. If yes, $v_i$ outputs $\dagpoa_i(B'.d)$ and $\dagdeli_i(B')$ right after. If no, $v_i$ uses the pull protocol to get \emph{all} missing blocks. In essence, $v_i$ must synchronize the whole causal history of $B'$ before outputting $\dagpoa$.
%
Upon outputting $\dagpoa$ for blocks in round $r{-}1$ from at least $2f{+}1$ validators, $v_i$ moves to round $r$ and creates a new block $B$ with \emph{all} these round $r{-}1$ blocks as $B.parents$. $v_i$ broadcasts $B$ in a best-effort way.

\para{Random pull.}
To pull a missing block $B$ in the \syncopt synchronizer protocol, the validator $v_i$ randomly chooses a constant set of validators $\mathcal{V}_B \subseteq \mathcal{V}$ and sends a request message to $\mathcal{V}_B$.
$v_i$ repeatedly sends the request message to different sets of validators at intervals until receiving $B$.
However, since $B$ is not certified and $v_i$ requests $B$ from randomly selected validators, there is no guarantee that $v_i$ will ever receive $B$.
After getting $B$, $v_i$ repeats the pull protocol to get all missing blocks of $causal(B)$. During the pull process, $v_i$ outputs $\dagpoa_i(B.d)$ and $\dagdeli_i(B)$ consecutively once it has synchronized $B$ and all blocks in $causal(B)$.

The \syncopt protocol achieves optimal push latency (G1), but not bounded amplification (G2).
Specifically, it achieves the round latency of $\delta$ with best-effort broadcast under optimistic cases. However, due to its random pull, it introduces uncertain pull latency and unbounded pull requests, thereby potentially unbounded push latency for every future round, despite having $O(1)$ pull complexity.

\section{Pull Induction Attacks and Key Insights}\label{sec-attack-insights}
Despite the rich design space proposed by prior work, none of them have taken a principled approach and achieved all ideal goals. This allows us to exploit their vulnerabilities through a new class of attacks we call pull induction attacks, which deliberately trigger unnecessary pulls to degrade their performance.
This section sheds light on the pain points of existing protocols and provides several key insights that guide the design of \sysname.

\subsection{Pull Induction Attacks}\label{sec-attack}
The goal of pull induction attacks is to induce honest validators to pull blocks from others, thereby increasing the push latency. To this end, the adversary selectively disseminates its blocks to only a subset of honest validators. Consequently, validators that do not receive the adversary's blocks are compelled to pull the missing blocks that are included by the protocol. \Cref{tab:comparison} (adverse case column) presents a performance comparison of different block synchronizer protocols under pull induction attacks.

We give a concrete example of a pull induction attack against the \syncopt synchronizer protocol adopted by Mysticeti and leave the discussion of pull induction attacks against other synchronizer protocols in the full paper~\cite[Appendix H]{kichidis2025beluga}.
In \Cref{fig:pull-induction-attack}, there are four validators $v_1$, $v_2$, $v_3$, and $v_4$, of which $v_4$ is the adversary. In round $r{-}1$, $v_4$ only disseminates its round $r{-}1$ block $B_4^{r{-}1}$ to $v_1$, making $v_1$'s round $r$ block $B_1^r$ reference $B_4^{r{-}1}$. When pushing round $r$ blocks and receiving $B_1^r$, both $v_2$ and $v_3$ miss $B_4^{r{-}1}$, and thus, they invoke the random pull protocol to fetch $B_4^{r{-}1}$ before accepting $B_1^r$ and having enough (i.e., 3 with $n{=}4$ and $f{=}1$) accepted round $r$ blocks to propose their round $r{+}1$ blocks. The latency of round $r$ thus consists of $\delta$ for pushing round $r$ blocks and at least $2\delta$ for pulling $B_4^{r{-}1}$. Similarly, in round $r$, $v_4$ only disseminates its round $r$ block $B_4^r$ to $v_2$, making $v_2$'s round $r{+}1$ block $B_2^{r{+}1}$ reference $B_4^r$. This will induce both $v_1$ and $v_3$ to pull $B_4^r$ before accepting $B_2^{r{+}1}$, thereby increasing the latency of round $r{+}1$ to at least $3\delta$. In round $r{+}1$, $v_4$ only disseminates its round $r{+}1$ block $B_4^{r{+}1}$ to $v_3$, which will induce both $v_1$ and $v_2$ to pull $B_4^{r{+}1}$ before accepting $B_3^{r{+}2}$, thereby increasing the latency of round $r{+}2$ to at least $3\delta$.
Recall that the consensus of Mysticeti requires three rounds of blocks~\cite[Algorithm 3]{babel2025mysticeti}. As a result, the consensus latency of Mysticeti is at least $9\delta$ under this pull induction attack.

\begin{figure}[t]
    \centering
    \includegraphics[width=0.8\linewidth]{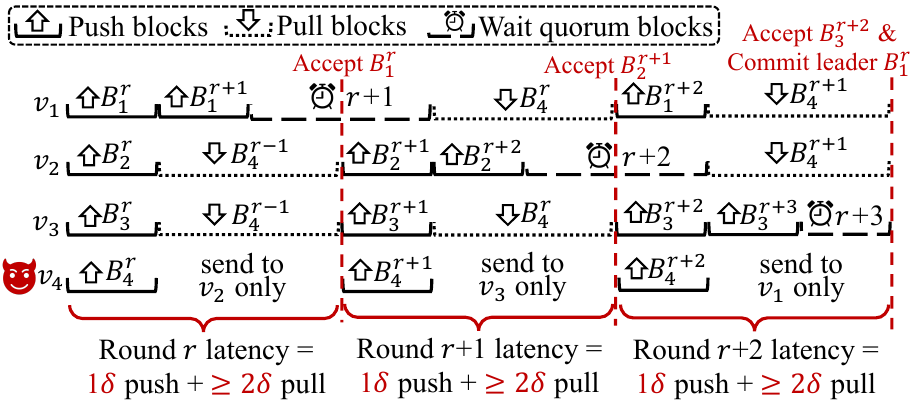}
    \caption{The pull induction attack: by selectively sharing its blocks (e.g.\ $v_4$ sends $B_4^{r{-}1}$ only to $v_1$, so $B_1^r$ references it), the adversary forces honest validators to trigger pulls to fetch the referenced parents, adding at least one pull round-trip of latency per round.}
    \label{fig:pull-induction-attack}
\end{figure}

\subsection{Key Insights}\label{sec-insights}
After reviewing existing block synchronizer protocols and their behavior under pull induction attacks, we have the following insights.

\para{Key insight 1: DAG-based synchronizer protocols can accelerate the consensus latency, but at risk of being delayed by the adversary.} In DAG-based synchronizer protocols, each block references at least $2f{+}1$ parents, and the connections between blocks can serve as proposal votes. This allows validators to complete the consensus during the formation of the DAG without extra communication by interpreting the DAG structure locally. As a result, the consensus latency of the protocols built on DAG-based synchronizer is soundly low under happy cases (e.g., $3\delta$ in Mysticeti). In contrast, the consensus protocols built on \syncmt must rely on a dependent consensus process to order blocks, which introduces extra communication overhead. For instance, apart from $2\delta$ push latency for pushing a round of blocks, Autobahn requires an additional $5\delta$ network latency to reach a consensus on a leader block, leading to the consensus latency of $7\delta$.

However, forcing a block to reference $2f{+}1$ parents enables the adversary to delay the progress of the synchronizer protocol with the pull induction attacks. Specifically, when an honest validator has its block reference some adversary blocks that are not shared with other honest validators, these honest validators must pull the missing adversary blocks before accepting the honest validator's block and collecting enough $2f{+}1$ accepted blocks to move to the next round. This increases the push latency by at least one pull round-trip. As a result, both push latency and consensus latency increase under the pull induction attacks.


\para{Key insight 2: Block certificates allow performing the push and pull protocols separately.} In the \syncmt synchronizer protocol, validators push each block accompanied by a certificate containing a quorum of $\geq f{+}1$ signatures, thereby attesting the block and causal availability. In this case, validators can output $\dagpoa$ for each certified block they receive without pulling any missing blocks in the block's causal history while still ensuring the Causal availability property. This enables the pull protocol to be performed in separation from the push protocol. We call this feature pulling blocks \emph{off the push path}, as pulling blocks can be completed independently from the push protocol. In contrast, blocks are pulled \emph{on the push path} if pulling missing blocks is necessary before validators accept them.

Supporting pulling blocks off the push path prevents Byzantine validators from proactively delaying the progress of the synchronizer protocol via pull induction attacks, since pulling missing blocks does not prevent validators from creating new blocks. However, the \syncmt synchronizer adopts an explicit certificate mechanism, where certificates are created within at least one round-trip latency, leading to higher push and consensus latencies compared to the \syncopt synchronizer protocol.



\para{Key insight 3: The pull protocol introduces a trade-off between push latency and communication complexity.} In certified synchronizer protocols (e.g., \synccert and \syncmt), the pull is deterministic, where validators pull missing blocks from a specific set of validators (with the set size $O(n)$). This ensures validators can fetch the missing blocks within a constant round trip (i.e., $2\delta$ with one for sending requests and the other for receiving blocks). However, this also introduces high communication complexity per validator (i.e., $O(n)$) since each validator might receive redundant blocks from others.

In contrast, the \syncopt synchronizer adopts a random pull protocol, where validators randomly pull the missing blocks from a small set of validators (with the set size $O(1)$). This reduces the communication complexity per validator to $O(1)$. However, such a random pull cannot ensure a constant round trip for synchronizing missing blocks. This is not a problem when the protocol is under happy cases, as pulling blocks does not impede progress. However, under adverse cases, the \syncopt synchronizer protocol requires at least $2\delta$ in the pull protocol, leading to at least $3\delta$ push latency and at least $9\delta$ consensus latency.

\section{The \sysname Protocol}\label{sec-sysname}
\subsection{Overview} \label{sec-sysname-overview}
\sysname is an efficient and robust DAG-based block synchronizer protocol composed of two key components: an AC-based optimistic push protocol (\autoref{sec-sysname-push}) and a hybrid pull protocol (\autoref{sec-sysname-pull}) based on the novel idea of Implicit Proof-of-Availability (ImPoA).

\para{AC-based optimistic push.}
Motivated by Insight 1, \sysname adopts a DAG structure and employs an optimistic push protocol, like \syncopt synchronizer, to achieve (G1) optimal push latency under happy cases and (G2) bounded retransmission. Specifically, validators disseminate blocks using a best-effort broadcast (that results in $\delta$ push latency) and stop retransmission once advancing to a new round.
However, unlike the DAG-based synchronizer protocols, where validators arbitrarily reference parent blocks, \sysname introduces an \emph{admission control (AC)} mechanism to filter out blocks based on the creators' behaviors. With AC, honest validators avoid referencing parent blocks that are
created by suspected Byzantine validators to trigger the pull protocol.
This mechanism effectively safeguards \sysname against pull induction attacks.

\para{ImPoA-based hybrid pull.}
Motivated by Insights 2 and 3, \sysname's pull protocol aims to separate pulling from pushing and balance pull complexity and latency. To this end, \sysname introduces an Implicit Proof-of-Availability (ImPoA) mechanism, which enables validators to identify blocks that can be safely accepted even if some of their ancestors are temporarily unavailable, thereby enabling validators to pull blocks off the push path. Based on ImPoA, we categorize pulling blocks into two types: \emph{live blocks} and \emph{bulk blocks}. Live blocks include missing ancestors that are not proven available, while bulk blocks only include missing ancestors that are proven available. Leveraging this distinction, \sysname employs a hybrid pull strategy: live blocks are pulled deterministically to minimize the latency, while bulk blocks are pulled in a randomized rotating way to reduce pull complexity. This allows \sysname to achieve (G2) bounded pull latency.

\para{\sysname Workflow.}
In \sysname, each validator $v_i$ pushes a block $B$ to the system by invoking $\dagbc_i(B,r)$ every round $r$, which consists of parent selection and broadcast.
Specifically, $v_i$ becomes ready to push a block for round $r$ after outputting $\dagpoa$ for blocks in round $r{-}1$ from at least $2f{+}1$ distinct validators (including itself). $v_i$ outputs $\dagpoa_i(B'.d)$ for a received block $B'$ when, for each parent, it either has received it and all its ancestors, or there exists an implicit certificate in $v_i$'s local view (i.e., the parent is referenced by $\geq f{+}1$ blocks; see Sec.~\ref{sec-sysname-pull-impl-poa}). After that, $v_i$ selects as parents a set of blocks that have been output via $\dagpoa$, subject to the admission control rules, and constructs a new block $B$ for round $r$. $B$ is then disseminated using a best-effort broadcast. Concurrently, $v_i$ executes the ImPoA-based hybrid pull protocol to fetch missing blocks. $v_i$ outputs $\dagdeli_i(B)$ to add $B$ into the DAG once it has received $B$ and output $\dagpoa_i(B.d)$.

\para{Block structure.}
To capture validators' behaviors during the block pushing process and facilitate the pull process, \sysname augments the block structure with three additional fields.
\begin{itemize}
    \item \emph{Weak links}: The $weaklinks$ field references blocks that a validator has received and accepted but not selected as parents. We therefore called $parents$ as strong links, both are used interchangeably throughout the paper.
    \item \emph{Watermark}: The $watermark$ is an $n$-element array maintained by each validator, where the $i$-th entry implies the highest round number of block received from validator $v_i$.
    \item \emph{Ancestors}: The $ancestors$ is an $n$-element array storing, for each validator $v_i$, the highest round number of $v_i$'s blocks reachable from the current block.
\end{itemize}
We note that the $weaklinks$ field has also been adopted by prior DAG-based synchronizer protocols, and thus, \sysname’s additional metadata mainly comes from the two $n$-element arrays, which add $16n$ bytes per block assuming 8-byte round numbers. 
This overhead is small: with $n=50$, the arrays require only 0.8 KB. In comparison, under our experiment (\autoref{sec-experiment}) of 300K TPS, 200 ms push latency, and 512-byte transactions, each block contains about 1,200 transactions, or roughly 0.6 MB of payload. Therefore, the metadata overhead is negligible compared with the block payload.


\subsection{AC-based Optimistic Push Protocol} \label{sec-sysname-push}
\sysname's push protocol specifies how validators create blocks and disseminate their created blocks to others. Block dissemination in \sysname is optimistic and relies on a best-effort broadcast; that is, validators disseminate blocks to all others without waiting for acknowledgments. Block creation is governed by an \emph{Admission Control} (AC) module, which enforces rules that filter blocks according to the creators' behaviors as quantified by a \emph{reputation} mechanism.
We detail the reputation mechanism and the AC module below, and defer the pseudo-code to the full paper~\cite[Appendix E]{kichidis2025beluga}.

\para{Reputation Mechanism.}
Each validator $v_i$ maintains a local reputation table $TR_i$ that records the reputations of all validators and is initialized with 0s. The reputation entry $TR_i[j]$ reflects the contribution of validator $v_j$ to the block pushing process. Specifically,
\begin{itemize}
    \item \emph{Reputation Increase:}
          $v_i$ increases $v_j$'s reputation if it receives $2f{+}1$ blocks (denoted by $\mathcal{B}_r$) in round $r$ that collectively indicate $v_j$'s round $r{-}1$ block has been received. Formally, if for each block $B\in\mathcal{B}_r$, $B.watermark[j] = r{-}1$, then $TR_i[j]$ is incremented.

    \item \emph{Reputation Decrease:}
          $v_i$ decreases $v_j$'s reputation whenever (i) it invokes the pull protocol to fetch a missing block created by $v_j$, or (ii) it receives pull requests for a $v_j$'s block from at least $f{+}1$ distinct validators. A pull request serves as a \emph{report} of $v_j$'s delayed dissemination behavior. Collecting $f{+}1$ such reports constitutes a \emph{blame} against $v_j$, indicating that at least one honest validator failed to receive $v_j$’s block in a timely manner. Thus, $v_i$ will decrease $v_j$'s reputation with a blame.
\end{itemize}
\emph{Intuition behind the reputation mechanism.}
The reputation mechanism is designed to capture validators' behaviors during the block pushing process after GST. Specifically, if a validator $v_i$ consistently pushes its blocks to all other validators, honest validators will frequently observe $2f{+}1$ blocks, indicating that $v_i$'s latest block is disseminated and received timely, thereby increasing $v_i$'s reputation. In contrast, if $v_i$ selectively pushes its blocks to only a subset of validators—performing a pull induction attack that forces others to invoke the pull protocol to retrieve its blocks—its reputation will decrease. Consequently, after GST, honest validators naturally maintain high reputations, whereas malicious validators that frequently launch such attacks accumulate low reputations.

In \sysname, the reputation increase is set by 1, while the reputation decrease is set by a large value $R_L$ (e.g., $R_L{=}10{,}000$). This asymmetry minimizes the number of rounds that the adversary can delay and eventually enables \sysname to achieve decent push and consensus latencies. A detailed analysis is presented in~\cite[Appendix C]{kichidis2025beluga}.

\para{Admission Control.}
The AC module determines which blocks are selected as parents for newly created blocks based on their creators' reputations.
Upon advancing to round $r$ and creating a new block, $v_i$ first collects the latest blocks it has received from all validators with round numbers $\le r{-}1$, denoted by the set $\mathcal{B}^{r{-}1}$. It then selects parent blocks from $\mathcal{B}^{r{-}1}$ through the following steps: (i) filter out any block in $\mathcal{B}^{r{-}1}$ that is not from round $r{-}1$ or is deemed unacceptable (i.e., $v_i$ has not yet output $\dagpoa$ for these blocks); (ii) from the remaining set, select the top $2f{+}1$ blocks in based on their creators' reputations in $TR_i$. A block is considered \emph{acceptable} if $v_i$ has received all of its ancestors or can otherwise ensure its ancestors are available (see Sec.~\ref{sec-sysname-pull-impl-poa} for more details). This AC mechanism ensures that honest validators avoid referencing blocks created by suspected Byzantine validators---specifically, those with low reputations due to previously triggering pull requests. By excluding such malicious blocks, an honest validator can create new blocks whose ancestors have already been received and accepted by all other honest validators, without invoking the pull protocol during the push process. Consequently, \sysname is inherently protected against pull induction attacks.

Apart from parents, a validator also references other acceptable blocks in $\mathcal{B}^{r{-}1}$ as $weaklinks$ in the new block. The $weaklinks$ serve as evidence of block availability to facilitate the pull process. We discuss it in \autoref{sec-sysname-pull}.

\para{Discussion.}
{
    The AC-based optimistic push protocol faces a practical challenge when validators advance rounds.
    On the one hand, due to the asymmetric connectivity, high-reputation blocks may not be received first by honest validators. If validators immediately reference the first received (potentially low-reputation) blocks as parents, the protocol may violate its intended guarantees.
    On the other hand, during asynchrony, reputation updates can be temporarily inaccurate, causing honest validators to appear lower-ranked than malicious ones. If validators wait for high-reputation blocks, malicious validators can delay progress arbitrarily or even stall the protocol by withholding their blocks.
}

{
    \sysname addresses this challenge by introducing a reputation threshold $R_t$, a per-round timeout $T_{rd}$, and and periodic reputation resets. Specifically, a validator $v_i$ advances to round $r$ with at least $2f{+}1$ round $r{-}1$ blocks that have been output via $\dagpoa_i$ and if one of the following conditions becomes satisfied:
    \begin{itemize}
        \item $v_i$ receives $2f{+}1$ blocks from round $r{-}1$ whose creators have reputations above a threshold $R_t{=}R_{2f+1}{-}R_L$, where $R_{2f+1}$ is the $2f{+}1$-th highest reputation in $TR_i$.
        \item $v_i$ is in round $r{-}1$, and the per-round timeout $T_{rd}$ expires, which is set to $5\Delta$ to ensure all honest blocks are received (\Cref{lem-round-liveness}).
        \item $v_i$ is in round ${<}r$ and observes some blocks of round ${>}r$.
    \end{itemize} 
    In brief, after GST and a reputation reset, these allow $v_i$ to advance to new rounds using either honest blocks or blocks that do not hinder protocol progress.
}

\subsection{ImPoA-based Hybrid Pull Protocol}\label{sec-sysname-pull}
\sysname's pull protocol defines how validators fetch missing blocks when they cannot be accepted during the push process. \sysname adopts an implicit proof-of-availability (ImPoA)-based hybrid pull protocol. As illustrated in \Cref{fig:impoa}, it comprises two components: an ImPoA-based pull mechanism and a hybrid pull strategy.

\begin{figure*}
    \centering
    \includegraphics[width=0.9\linewidth]{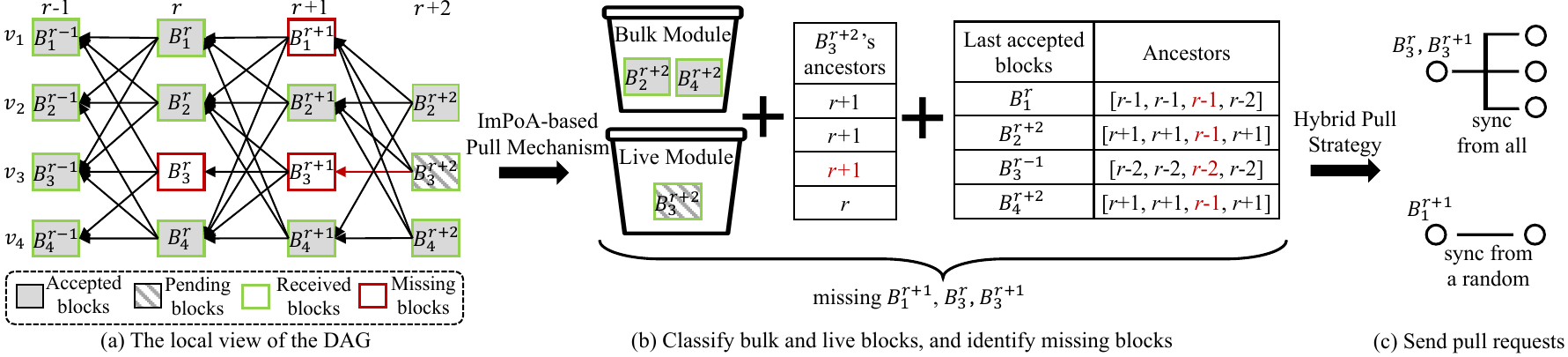}
    \caption{The ImPoA-based hybrid pull protocol for $v_4$: (a) $v_4$ accepts $B_2^{r{+}2}$ and $B_4^{r{+}2}$ whose parent $B_1^{r{+}1}$ is only implicitly available (both reference it); (b) it advertises the blocks referencing missing ones to the bulk and live synchronizer modules, identifying the blocks to fetch; (c) it fetches them via a hybrid pull that balances latency and complexity.}
    \label{fig:impoa}
\end{figure*}

\subsubsection{ImPoA-based Pull Mechanism}\label{sec-sysname-pull-impl-poa}
The ImPoA-based pull mechanism enables validators to pull certain blocks off the push path once their availability can be proven. However, unlike the existing certified synchronizer protocols (such as \syncmt) that create \emph{explicit} block certificates at the cost of additional communication overhead during the push process, \sysname constructs \emph{implicit} block certificates to achieve the same availability guarantees by locally interpreting block patterns.

\para{Implicit PoA.}
In \sysname, a block $B$ is identified as \emph{implicitly available} if it is referenced (via strong or weak links) by at least $f{+}1$ blocks from the subsequent rounds. Note that a validator references $B$ only if it (i) receives $B$, and (ii) can verify the availability of $B$'s causal history. Consequently, these $f{+}1$ referencing blocks collectively form an implicit proof-of-availability (PoA) for $B$, implying that at least one honest validator receives $B$ and for every ancestor of $B$, at least one honest validator has received it. 
For instance, in \Cref{fig:impoa}(a), the missing block $B_1^{r{+}1}$ is identified implicitly available as it was referenced by two received blocks $B_2^{r{+}2}$ and $B_4^{r{+}2}$.

\para{Live and Bulk modules.}
With ImPoA, validators can safely accept blocks even when parts of their causal histories are missing, thereby allowing missing blocks to be pulled off the push path. To accommodate the mechanism, \sysname employs two pull synchronizer modules: (i) \emph{bulk synchronizer module}, which retrieves missing blocks off the push path, and (ii) \emph{live synchronizer module}, which retrieves missing blocks on the push path.

\para{Mechanism specification.}
We now describe the workflow of our pull mechanism. Assume a validator $v_i$ is currently proceeding in round $r$.
Upon receiving a block $B$ containing missing parents during the push process, $v_i$ determines which pull synchronizer module will be used to pull $B$'s missing ancestors based on its local view. We denote $\mathcal{B}_{bk}$ the block set in the bulk synchronizer module and $\mathcal{B}_{lv}$ the block set in the live synchronizer module.



Specifically, $v_i$ first checks whether $B$ can be proven available. If every parent of $B$ has an implicit PoA from $v_i$'s received blocks or has been $\dagpoa$, then $v_i$ ensures $B$ is available. $v_i$ then transmits $B$ to the bulk synchronizer, outputs $\dagpoa$ and $\dagdeli$ for $B$, and include $B$ into the bulk block set $\mathcal{B}_{bk}$. If, otherwise, $B$ is not proven available; $v_i$ then transmits $B$ to the live synchronizer module and includes $B$ into $\mathcal{B}_{lv}$. For example, in \Cref{fig:impoa}(a), $v_4$ is proceeding in round $r{+}2$ and transmits $B_2^{r{+}2}$ and $B_4^{r{+}2}$ to the bulk synchronizer module even though their common parent $B_1^{r{+}1}$ was missing, since $B_1^{r{+}1}$ has an implicit PoA. In contrast, $B_3^{r{+}2}$ is transmitted to the live synchronizer module.

Moreover, $v_i$ can dynamically transmit blocks from $\mathcal{B}_{lv}$ to $\mathcal{B}_{bk}$. In particular, upon receiving a new block, $v_i$ checks whether any block $B' \in \mathcal{B}_{lv}$ is proven available. If yes, $v_i$ transmits such available $B'$ from $\mathcal{B}_{lv}$ to $\mathcal{B}_{bk}$.


\para{Identifying missing blocks.} By utilizing the ImPoA-based pull mechanism, \sysname classifies blocks with missing ancestors into two categories: \emph{Live blocks} (i.e., blocks in $\mathcal{B}_{lv}$) and \emph{Bulk blocks} (i.e., blocks in $\mathcal{B}_{bk}$).
Each validator $v_i$ can then locally identify missing blocks from $\mathcal{B}_{lv}$ and $\mathcal{B}_{bk}$ by checking their $ancestors$.
However, since blocks in the DAG are well-connected, $\mathcal{B}_{lv}$ and $\mathcal{B}_{bk}$ might involve overlapped missing ancestors. For instance, in \Cref{fig:impoa}, the missing block $B_1^{r{+}1}$ is the common ancestor of both the live blocks and the bulk blocks. As a result, $v_i$ might pull the same missing blocks in both the live and bulk synchronizer modules redundantly.

To avoid redundant synchronization, $v_i$ traces the last accepted blocks from all validators and combines them with blocks' \emph{ancestors} to identify missing blocks for distinct synchronizer modules.
Since each block must reference the previous block proposed by the same validator~\cite{babel2025mysticeti}, the last accepted block $B_j^r$ from $v_j$ in round $r$ indicates that all blocks from $v_j$ with a round $r'\leq r$ are proven available and can be accepted by $v_i$. Moreover, if a validator accepts a block, it means that the block's causal history is available and can be accepted.
As a result, with the information, $v_i$ can identify the missing blocks that are \emph{only} required by the live blocks.

As shown in \Cref{fig:impoa}(b), the last accepted blocks and their $ancestors$ information imply that $v_4$ (i) will be able to accept $r{+}1$ block $B_1^{r{+}1}$ from $v_1$, and (ii) has accepted round $r{+}2$ block $B_2^{r{+}2}$ from $v_2$, round $r{-}1$ block $B_3^{r{-}1}$ from $v_3$, and round $r{+}2$ block $B_4^{r{+}2}$ from itself. When processing the pending block $B_3^{r{+}2}$ (which references $v_3$'s round $r{+}1$ block) in the live synchronizer module, $v_4$ can identify $B_3^r$ and $B_3^{r{+}1}$ are only required by the live blocks (i.e., $B_3^{r{+}2}$) but $B_1^{r{+}1}$ does not. As a result, the live synchronizer module only requests the missing blocks $B_3^r$ and $B_3^{r{+}1}$, while the bulk synchronizer module requests the missing block $B_1^{r{+}1}$.

\subsubsection{Hybrid Pull Strategy}\label{sec-sysname-pull-strategy}
After identifying the missing blocks that a validator $v_i$ needs to fetch, $v_i$ deploys a hybrid pull strategy to balance message and round complexities, as shown in \Cref{fig:impoa}(c).

\para{Pulling blocks in the live synchronizer.}
Blocks $\mathcal{B}_{lv}$ in the live synchronizer module are time-sensitive, as their missing blocks can block the push process. Thus, the live synchronizer module adopts a \emph{deterministic pull strategy} to minimize the latency. Specifically, $v_i$ sends the pull requests for all missing blocks specified in $\mathcal{B}_{lv}$ to all validators. Such a pull strategy might pull redundant blocks. However, it guarantees that $v_i$ can receive missing blocks as long as they are $\dagpoa ed$ by one honest validator, and process live blocks within a round-trip delay (i.e., $2\delta$).

\para{Pulling blocks in the bulk synchronizer.}
Pulling blocks in the bulk synchronizer does not block the push process, and thus, is not sensitive to latency. Consequently, the bulk synchronizer module adopts a \emph{randomized rotating pull strategy} to minimize the pull complexity. Specifically, for each missing block specified in $\mathcal{B}_{bk}$, $v_i$ randomly chooses a small set of validators (with the set size of $O(1)$) to send the pull request. If $v_i$ does not receive the requested missing block within some predefined time (say $\Delta_{bk}$), it randomly chooses another set of validators who have not been requested by $v_i$ to send the pull request.

According to the definitions and pull strategies, it is obvious that all honest blocks can become acceptable within a round-trip delay. After getting missing blocks, $v_i$ outputs $\dagpoa$ and $\dagdeli$ for each of them if $v_i$ hasn't done it before.

\subsection{Building BFT Consensus on \sysname} \label{sec-sysname-consensus}
\sysname is inspired by modern multi-proposer BFT protocols and can serve as a modular block synchronization layer for a broad class of partially synchronous BFT consensus protocols.
As illustrated in \Cref{fig:overview}, \sysname interfaces with a generic state machine replication system by supplying accepted blocks (i.e., those that have been output via $\dagpoa$) to the consensus layer for ordering and supplying available blocks (i.e., those that have been output via $\dagdeli$) to the execution layer for execution after these blocks are ordered.
The key enabler of the generality is that the DAG structure constructed by \sysname induces structural \emph{patterns} on blocks which can be used by different consensus protocols to implement their ordering logic.

\para{Block patterns.}
{
    We define two patterns that arise naturally in \sysname. A block forms an \emph{availability pattern} if it is referenced by more than $f$ blocks produced by distinct validators.
    A block forms a \emph{certificate pattern} if it is referenced as parents by more than $2f$ blocks from distinct validators.
    Intuitively, the availability pattern ensures that the block and its causal history can be retrieved by honest validators; while the certification pattern implies uniqueness: for any validator and round, at most one block can become certified (note that an honest validator creates at most one block per round). A block that forms an availability pattern (resp., certificate pattern) is called an \emph{available} block (resp., \emph{certified} block). Both patterns involve two push rounds to generate. Moreover, we say that a block \emph{votes for} a certified block $B_C$ if it references at least $2f{+}1$ blocks as parents that themselves reference $B_C$.
}


\para{Applying consensus on \sysname.}
These patterns enable BFT protocols to operate over the DAG constructed by \sysname while applying their own ordering rules. We illustrate this integration using the protocols in \Cref{tab:comparison}, and summarize the resulting performance (shown in \Cref{tab:comparison}).

For multi-chain protocols such as Dashing~\cite{duan2024dashing} and Autobahn~\cite{giridharan2024autobahn}, validators execute separate consensus instances to ensure non-equivocation and thus only rely on \sysname to ensure block availability. They use availability patterns when performing separate consensus instances. Specifically, a leader validator uses its most recent block that forms an availability pattern as a proposal to coordinate a consensus instance. This design achieves improved push latency in optimistic conditions, at the cost of a modest increase in consensus latency under adverse conditions.

DAG-based protocols perform consensus by locally interpreting the DAG structure. When instantiated over \sysname, they rely on certificate patterns to ensure non-equivocation.
In these protocols, certain blocks are designated as leader blocks to drive consensus progress. During the push phase, \sysname's admission control prioritizes the inclusion of leader blocks and their supporting blocks from validators with benign reputations.
Validators then apply protocol-specific rules to order blocks.

Protocols such as Mysticeti~\cite{babel2025mysticeti} and Cordial Miners~\cite{cordial-miners} employ the same DAG structure as \sysname (i.e., constructing a DAG using a best-effort broadcast approach but distinct admission control and pull mechanisms). Therefore, their ordering rules can be applied to \sysname directly to derive a BFT consensus protocol.
The full paper~\cite[Appendix D]{kichidis2025beluga} details the integration of \sysname into Mysticeti and presents formal security proofs.
These consensus protocols achieve the same optimal push and consensus latencies in optimistic scenarios when instantiated over \sysname, while offering greater robustness under adverse conditions. 

Protocols such as Bullshark~\cite{spiegelman2022bullshark} and successor Shoal~\cite{spiegelman2025shoal} commit a leader certified block once it receives sufficient votes from subsequent certified blocks. When executed over \sysname, this process requires two rounds of certified blocks, resulting in a consensus latency of $4\delta$ in optimistic cases and $8\delta$ in adverse cases.

Protocols such as Shoal++~\cite{arun2024shoal++}, Sailfish~\cite{shrestha2024sailfish}, and Sailfish++~\cite{shrestha2025optimistic} commit a leader certified block using one round of certified blocks and the first messages of the next round. A leader certified block is directly committed once at least $2f{+}1$ first messages of the next-round certified blocks vote for it. In \sysname, these first messages correspond to individual blocks, allowing commitment with three push rounds. This yields a consensus latency of $3\delta$ in optimistic cases and $6\delta$ in adverse cases.



\para{Garbage collection.}
\sysname naturally reuses the garbage collection mechanism introduced by Narwhal~\cite{danezis2022narwhal}; this component allows the protocol to clean up any partially disseminated messages that were not promptly committed, preventing unbounded storage and memory growth. Notably, this module is employed in most implementations, both in DAG-based and linear BFT protocols~\cite{autobahn-code,narwhal-code,sui-code}.

\begin{figure}
    \centering
    \includegraphics[width=0.8\linewidth]{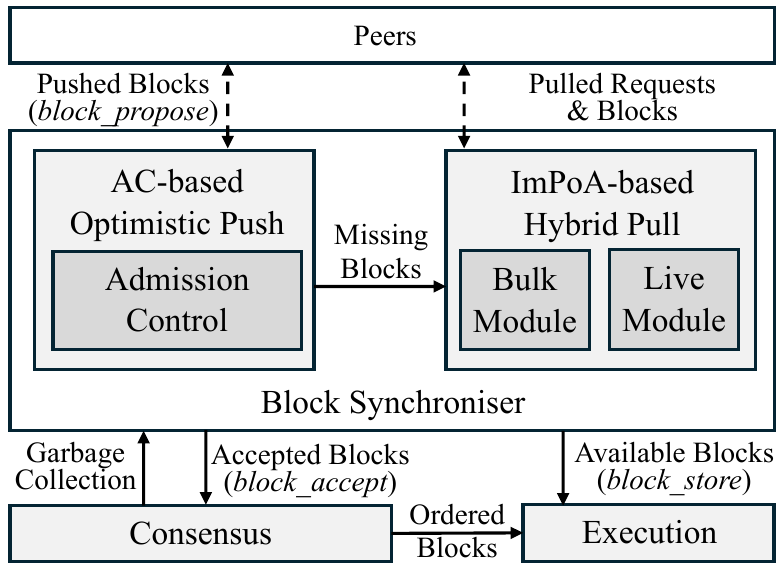}
    \caption{\sysname can be integrated into a generic BFT consensus protocol. By applying the consensus rules on the blocks produced by \sysname, validators derive a consistent order for blocks. \sysname guarantees the availability of ordered blocks for the state machine replication (SMR) execution.}
    \label{fig:overview}
    \Description{}
\end{figure}

\section{Protocol Analysis}  \label{sec-analysis} \label{sec-correctness-analysis}

We prove \sysname satisfies the properties defined in \Cref{def-dag-sync}. The performance analyses of \sysname in the happy case and under adversarial settings are deferred to the full paper~\cite[Appendix C]{kichidis2025beluga}. All theorems in this section have been formalised with their proofs machine-checked in Lean~\cite{Moura021}; the Lean code is provided in the accompanying material and outlined in~\cite[Appendix I]{kichidis2025beluga}

Like previous literature, we ignore the time spent handling (e.g., verifying and executing) a block in proofs, since it is negligible compared to the network delay $\Delta$.
Moreover, combining our push and pull protocols, it is intuitive that every block in the global pool (i.e., the block has been output via $\dagpoa$ by all honest validators) is eventually received by every honest validator.

\begin{theorem}[Block availability]\label{lem-block-availability}
    \sysname satisfies block availability. If an honest validator $v_i$ outputs $\dagpoa_i(B.d)$ for some block $B$ produced $r$, then $v_i$ eventually outputs $\dagdeli_i(B)$
\end{theorem}
\begin{proof}
    In \sysname, an honest validator $v_i$ outputs $\dagpoa_i(B.d)$ for some block $B$ in round $r$ when $v_i$ received $B$. As a result, $v_i$ must have stored $B$ and output $\dagdeli_i(B)$.
\end{proof}

\begin{theorem}[Causal availability]\label{lem-causal-availability}
    \sysname satisfies causal availability. If an honest validator $v_i$ outputs $\dagpoa_i(B.d)$ for some block $B$, then for every block $B' {\in} causal(B)$, $v_i$ eventually outputs $\dagpoa_i(B'.d)$, where $causal(B)$ represents $B$'s causal history.
\end{theorem}
\begin{proof}
    In \sysname, an honest validator $v_i$ outputs $\dagpoa_i(B.d)$ for some block $B$ in round $r$ when $v_i$ received $B$ and ensures all its parents in round $r{-}1$ are available---that is, $v_i$ has either outputted $\dagpoa$ for the parent blocks or observed they are referenced by at least $f{+}1$ subsequent blocks (\autoref{sec-sysname-pull}). In the latter case where the parent blocks (denoted by a set $\mathcal{B}^{r{-}1}$) are referenced by at least $f{+}1$ subsequent blocks, according to \sysname's push protocol, the creators of these $f{+}1$ subsequent blocks must have outputted $\dagpoa$ for $\mathcal{B}^{r{-}1}$. This means that at least one honest validator $v_j$ has stored $\mathcal{B}^{r{-}1}$ and ensured all parents of $\mathcal{B}^{r{-}1}$ in round $r{-}2$ are available. As a result, $v_j$ must have either outputted $\dagpoa$ for the parent blocks of $\mathcal{B}^{r{-}1}$ in round $r{-}2$ or observed they are referenced by at least $f{+}1$ subsequent blocks. By induction, we can see that for every block $B' \in causal(B)$, there exists at least one honest validator that has stored $B'$ and ensured all its parents are available. As a result, $v_i$ can eventually receive $B'$ from this honest validator via the pull protocol and output $\dagpoa_i(B'.d)$.
\end{proof}

\begin{lemma} \label{lem-round-liveness}
    After GST, if round $r$ is the highest round that honest validators are at some time $t$, then all honest validators will enter the round $r$ by $t+4\Delta$.
\end{lemma}
\begin{proof}
    Let $v_i$ be the honest validator in round $r$ at time $t$ and $\mathcal{V}_{slow}$ be a set of honest validators proceeding in round $r'<r{-}1$ at $t$. After GST, the message delay between honest validators is bounded by $\Delta$. Thus, each honest validator must receive the latest blocks from all honest validators by $t{+}\Delta$. With our pull protocol, honest validators can accept all these honest latest blocks via either block fetching or imPoAs by time $t{+}3\Delta$. Now consider two cases.
    
    \textit{Case 1:} $\mathcal{V}_{slow}=\emptyset$. This means that all honest validators are proceeding in round $\geq r{-}1$ by $t$. Then by time $t{+}3\Delta$, all validators have entered round $r$ and created their round $r$ blocks.

    \textit{Case 2:} $\mathcal{V}_{slow}\neq\emptyset$. According to \sysname's round advancement rules (\autoref{sec-sysname-push}), validators in $\mathcal{V}_{slow}$ advance to round $r-1$ immediately and create their round $r{-}1$ blocks by time $t{+}3\Delta$. Thus, by time $t{+}4\Delta$, each honest validator has received round $r{-}1$ blocks from all honest validators and enters round $r$.
\end{proof}

\begin{theorem}[Round-Progression]\label{thm-round-progression}
    \sysname satisfies round-progression. For each round $r\geq 0$, at least $2f{+}1$ validators will create and disseminate their round $r$ blocks.
\end{theorem}
\begin{proof}
    For the genesis round $0$, all validators will create and disseminate their round $0$ blocks. Thus, the lemma holds for round $0$.
    In addition, according to \Cref{lem-round-liveness}, all honest validators can enter the same round (w.l.o.g. at round $r$) within $4\Delta$ after GST. As a result, all (i.e., at least $2f{+}1$) honest validators must be able to create and disseminate their round $r$ blocks by time $GST{+}4\Delta$. Thanks to \sysname's pull protocol, every honest validator can accept at least $2f{+}1$ round $r$ blocks. With the round advancement rule, validators must advance to the next round $r{+}1$ within a timeout $T_{rd}$, and thus can create their round $r{+}1$ blocks. By induction on the round, we can see that for any future round $r'{>}r$, at least $2f{+}1$ validators will create and disseminate their round $r'$ blocks. Moreover, for any round $1 \leq r'' \leq r$, since \sysname's push protocol requires validators to reference at least $2f{+}1$ parent blocks from the previous round when creating their round $r''$ blocks, at least $2f{+}1$ validators must have created and disseminated their round $r''-1$ blocks. By induction, we can see that for any previous round $1 \leq r'' \leq r$, at least $2f{+}1$ validators must have created and disseminated their round $r''$ blocks. The proof is done.
\end{proof}

\begin{theorem}[Round-Termination]\label{thm-round-termination}
    \sysname satisfies round-termination. For each round $r\geq 0$, each honest validator accepts block proposals, whose assigned round $r$, from at least $2f{+}1$ different validators.
\end{theorem}
\begin{proof}
    For the genesis round $0$, since all honest validators create their round $0$ blocks, and round $0$ blocks do not reference any blocks, each honest validator can accept at least $2f{+}1$ round $0$ blocks. The lemma holds for round $0$.
    In addition, according to \Cref{thm-round-progression}, for each round $r\geq 1$, at least $2f{+}1$ validators will create and disseminate their round $r$ blocks. Each round $r$ blocks consist of at least $f{+}1$ blocks created by honest validators. For these $f{+}1$ honest validators, according to \sysname's push protocol, they must output $\dagpoa$ to accept at least $2f{+}1$ round $r{-}1$ blocks. According to \Cref{lem-block-availability} and \Cref{lem-causal-availability}, these round $r{-}1$ blocks and their causal histories are available to all honest validators. As a result, each honest validator can accept at least $2f{+}1$ round $r{-}1$ blocks. By induction, we can see that for any round $r\geq 1$, each honest validator can accept at least $2f{+}1$ round $r$ blocks.
\end{proof}
\section{Experimental Evaluation}\label{sec-experiment}

\begin{figure*}[t]
    \centering
    \includegraphics[width=0.9\textwidth]{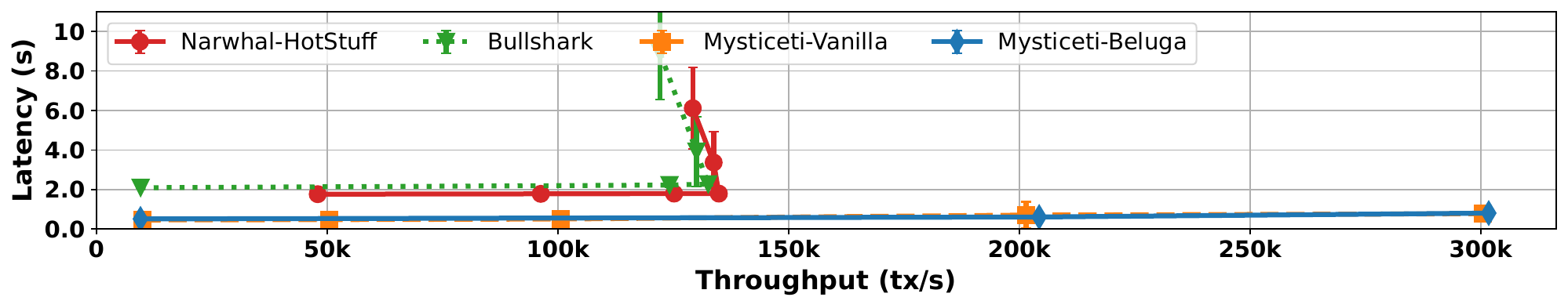}
    \caption{
        Throughput--latency of Mysticeti with \sysname vs.\ the baseline push-pull synchronizer. WAN, 50 validators, no faults, 512B transactions.
    }
    \label{fig:happy-path}
\end{figure*}

We then show its improvements over the baseline Mysticeti implemented with the default push-pull block synchronizer module (\autoref{sec-problem-def}).
We chose to implement \sysname within Mysticeti because (1) it is a state-of-the-art DAG-based BFT consensus protocol deployed in production in multiple blockchains~\cite{sui,iota-consensus}, providing real-world impact, (2) its open-source codebase~\cite{mysticeti-code} is well documented and modular, facilitating implementation, and (3) it builds upon an uncertified DAG, making it particularly vulnerable to the attacks described in \autoref{sec-attack}. These properties make Mysticeti an ideal candidate to demonstrate the effectiveness of \sysname.

Our evaluation demonstrates the following claims:
\begin{enumerate}[label={\bfseries C\arabic*}]
    \item\label{claim:c1} \sysname introduces no noticeable performance overhead when the protocol runs in ideal conditions (no Byzantine parties and synchronous network)
    \item\label{claim:c2} \sysname drastically improves both latency and throughput in the presence of asynchronous network conditions and the attack presented in \autoref{sec-attack}.
\end{enumerate}

Evaluating the performance of BFT protocols in the presence of generic Byzantine faults is an open research question~\cite{twins}, and state-of-the-art evidence relies on formal proofs.

\subsection{Benchmarks in Ideal Conditions}\label{sec:happy-path-benchmark}

\Cref{fig:happy-path} depicts the performance of Mysticeti equipped with both \sysname and the baseline block synchronizer operating through a best-effort push protocol followed by a random pull (described in \autoref{sec-problem-def}) running with 50 honest validators.
As expected, the performance of the baseline Mysticeti is similar to Mysticeti equipped with \sysname. Both systems exhibit a stable throughput up to around 100,000 tx/s while maintaining a latency of around 0.5s, and both scale easily to 300,000 tx/s (512-byte transactions) while maintaining sub-second latency (0.8s). We did not push throughput further for cost reasons, but both systems appear to use less than 10\% of their CPU at that point, indicating that they could likely handle even higher loads.
This performance similarity is due to \sysname's lightweight requirements during the happy path: when all parties are honest and the network is synchronous, \sysname imposes minimal constraints on parent block selection and thus operates similarly to the baseline push-pull synchronizer.
This confirms \Cref{claim:c1} that \sysname does not introduce any significant overhead when the network is synchronous and all parties are honest.
Additionally, \Cref{fig:happy-path} illustrates that \sysname allows to retain the superior performance of Mysticeti with respect to Bullshark~\cite{spiegelman2022bullshark} (an example of certified DAG-based protocol) and Narwhal-HotStuff~\cite{danezis2022narwhal} (an example with explicit separation between data dissemination and consensus).

\subsection{Benchmarks under Attack}\label{sec:attack-benchmark}

\Cref{fig:faults} depicts the performance of Mysticeti with both \sysname and the baseline block synchronizer in a 10-validator deployment with 1 or 3 faults (we limit this benchmark to 10 validators for cost reasons). The Byzantine validators perform the pull induction attack described in \autoref{sec-attack}, creating conditions that can also arise from severe network asynchrony.
The baseline Mysticeti suffers severe throughput degradation and dramatic latency increases. For reference, the graph includes the no-fault performance with 10 validators to illustrate the performance gap under attack. With three Byzantine faults, baseline Mysticeti's throughput drops by over 15x and its latency increases to over 50 seconds, compared to 0.5 seconds in the fault-free case. In contrast, Mysticeti equipped with \sysname maintains substantially higher throughput: \sysname rapidly detects and prioritizes inclusion of blocks from reliable validators in the DAG. The throughput reduction (30\%) is primarily attributable to the loss of faulty validator capacity. Latency increases by approximately 4x to around 2 seconds, which is expected given the need to wait for additional blocks before proceeding to the next round and to recover from attack-induced asynchrony.
This represents a substantial improvement over baseline Mysticeti: over 3x higher throughput and 25x lower latency with three Byzantine faults. These results confirm \Cref{claim:c2} that \sysname significantly improves both latency and throughput under asynchronous network conditions and pull induction attacks.

\begin{figure}[t]
    \centering
    \includegraphics[width=\columnwidth]{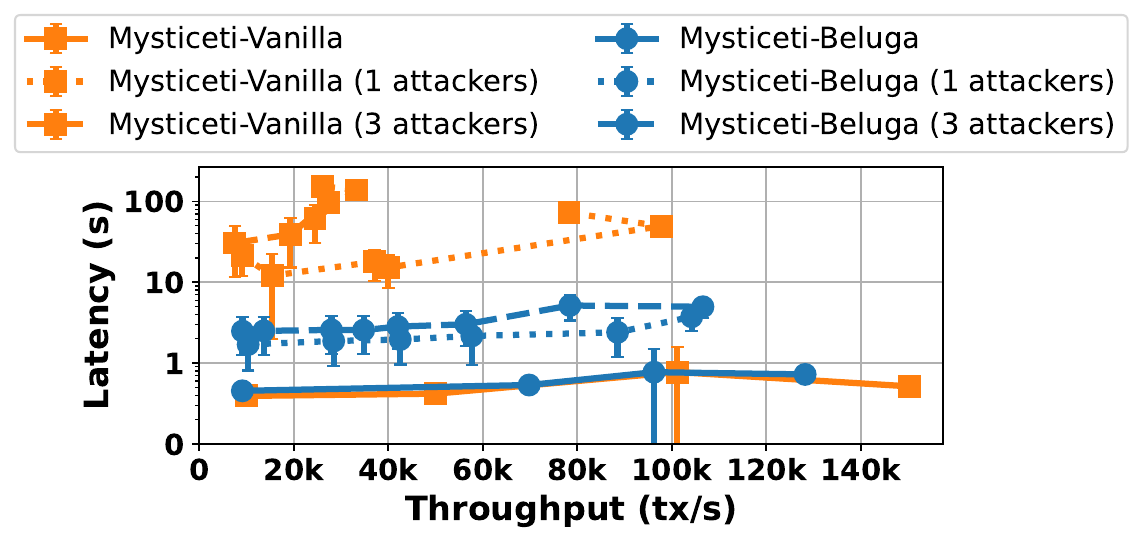}
    \caption{
        Throughput--latency of Mysticeti with \sysname vs.\ the baseline push-pull synchronizer under the pull induction attack. WAN, 10 validators, 0/1/3 faults, 512B transactions; log-scale latency.
    }
    \label{fig:faults}
    \Description{}
\end{figure}

\section{From Paper to Mainnet}

The motivation for developing \sysname arose in March 2024 following a testnet incident on the Sui blockchain, triggered by a fundamental change in the protocol's networking stack. Some validators misconfigured their machines, resulting in a situation where they could receive blocks from other validators but were unable to propagate their own. Consequently, these validators accumulated a large number of blocks that remained unshared with the network. In one extreme case, a single validator locally created approximately 750,000 blocks that had not been disseminated. Once the validator corrected their configuration and broadcast all these blocks simultaneously, the network experienced a stall: fast validators began including these blocks in their DAGs, and once $f+1$ of them did so, the rest of the network had to synchronize as well, causing a substantial performance degradation. Specifically, all 750,000 blocks were committed within one minute, triggering an overwhelming number of pull requests.

This incident revealed the root cause in the block synchronizer component. Its admission control module (see \Cref{fig:overview}) lacked sufficient filtering for poorly performing validators, allowing the most recent hoarded blocks to enter local DAGs. Moreover, the original pull protocol exacerbated the problem by indiscriminately pulling from both poorly performing and random validators, resulting in high synchronization latency. These observations underscored the need for a high-quality DAG that minimizes synchronization from underperforming validators.

We collaborated with the Sui team to experiment with various heuristics inspired by leader scoring~\cite{hammerhead}, but initially, excluding ancestors indiscriminately disrupted block propagation even under benign conditions. This motivated the introduction of optional dependencies (\autoref{sec-sysname}). Subsequently, they embarked on a multi-month effort to collect detailed consensus metrics and evaluate different strategies for scoring validators. This process ultimately led to the design presented in \Cref{sec-sysname}. After nearly a year of testing and tuning, \sysname was deployed on Sui mainnet version \texttt{mainnet-v1.42.0} in January 2025.

\Cref{fig:production-latency} (\autoref{sec:introduction}) shows a production network replicating the Sui mainnet with all 135 validators, in their respective geo-location and with distribution stake, sustaining a constant load of 6,000 transactions per second. The figure compares the original Mysticeti protocol run within the Sui mainnet before \sysname deployment, and the improved version with \sysname. The results demonstrate that \sysname effectively mitigates the latency spikes previously observed in the network under attack, resulting in a stable and predictable transaction commit latency. \sysname effectively prevents poorly performing validators from adversely impacting the overall network performance, resulting in a 5x reduction in the 95th percentile, a 2x reduction in the 75th percentile, and a 25\% reduction in the 50th percentile latency.

\Cref{fig:production-block-proposal} shows the rate at which blocks from an intentionally slow validator (green line) appear in the committed DAG under the same production setup after deploying \sysname. The validator's contribution to the committed DAG falls to zero within roughly 1-2 minutes as its score drops and the admission-control module at other validators discards its blocks. It remains zero throughout the attack, and then returns to its original rate within 3 minutes of recovering. Notice how other validators remain largely unaffected (other colored lined on the figure). This result demonstrates that \sysname rapidly identifies and deprioritizes poorly performing validators, preventing them from degrading network performance during misbehavior, and then promptly restores their ability to participate once they return to normal behavior. Thus, recovering validators rejoin the consensus process without long-term penalties.

\begin{figure}[t]
    \centering
    \includegraphics[width=0.9\columnwidth]{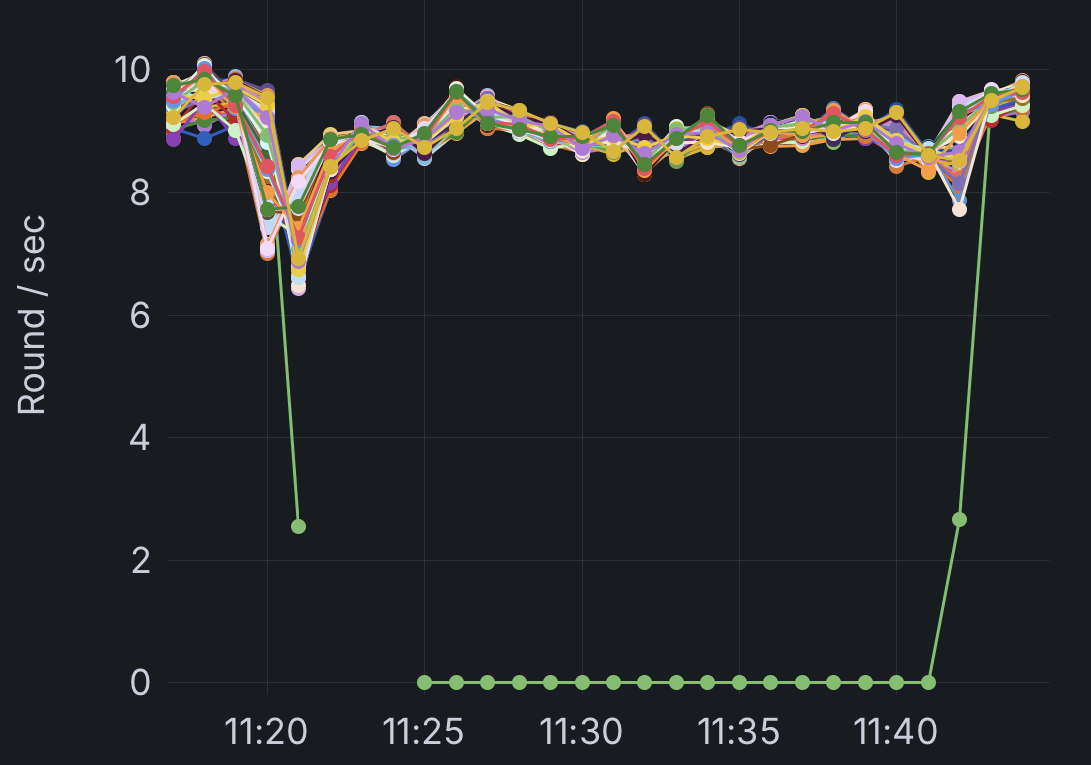}
    \caption{
        Per-validator block proposal rate in a Sui mainnet reproduction (135 validators) with \sysname. The slow validator's rate drops to zero during the attack and recovers once performance is restored.
    }
    \label{fig:production-block-proposal}
    \Description{}
\end{figure}



\section{Related Work}\label{sec-related-work}

A largely unexplored dimension of the consensus literature concerns the design of synchronization primitives themselves. While many protocols assume a specific synchrony model, prior work, to our knowledge, does not isolate, specify, or evaluate such primitives as first-class, modular components of consensus stacks. We address this gap by defining and implementing synchronization abstractions that compose cleanly with existing consensus protocols.

One key source of inspiration is the Pacemaker~\cite{baudet2019state} of leader-based protocols. The pacemaker separates view-change timing, leader rotation, and timeout management from the safety-critical core, which creates clean interfaces and lets liveness mechanisms evolve independently. Our synchronizer plays a similar role: it fits into the modular blockchain architecture~\cite{cohen2023proof}, defines isolated processes that deliver high-performance support to the safety-critical total-ordering algorithm, and lets engineers separate concerns.

A synchronizer protocol can be viewed as a module to ensure totality of reliable broadcast~\cite{rabin1989efficient}, one of the most studied problems in distributed computing. Unlike AVID~\cite{cachin2005asynchronous,alhaddad2022asynchronous} protocols, which focus on communication complexity in asynchronous settings, or optimistic reliable broadcast~\cite{shrestha2025optimistic}, which emphasizes latency under a low number of faults, our design examines performance after GST in partially synchronous or synchronous networks. This approach provides a more practical perspective and has already enabled significant performance improvements in Sui, a top-20 blockchain.

In permissionless networks, GossipSub~\cite{vyzovitis2020gossipsub} addresses related concerns at the gossip overlay layer. Like \sysname, it decomposes message propagation into an eager push phase (via a structured local mesh) followed by a lazy pull phase (gossip-driven requests for missing messages), and it incorporates a peer scoring function that penalizes misbehaving nodes to maintain mesh quality. However, GossipSub targets open, Sybil-prone environments (Filecoin and Ethereum~2.0) and operates purely at the network transport layer, without awareness of causal block structure or consensus semantics. In contrast, \sysname is a consensus-aware module that exploits DAG structure for scarcity detection and provides provable reliable broadcast guarantees.

\ifpublish
    Starfish~\cite{polyanskii2025starfish} decouples availability from ordering, letting validators vote for blocks whose causal history they have not yet received and relying on a separate availability layer to filter unavailable blocks before commit. This removes pull synchronization from the critical path but can lower commit rates, waste bandwidth, and requires every validator to produce a block every round. \sysname instead guarantees availability inline via its admission control and pull protocol, avoiding a separate availability layer while preserving optimistic push performance.
\fi


\ifpublish
    \begin{acks}
    This work is funded by Mysten Labs and was conducted while Jianting Zhang was interning with the company. We thank George Danezis, Adrian Perrig, and Philipp Jovanovic for their insightful discussions that greatly improved this work.
\end{acks}

\fi

\bibliographystyle{ACM-Reference-Format}
\bibliography{references}

@book{cachin2011introduction,
    title     = {Introduction to reliable and secure distributed programming},
    author    = {Cachin, Christian and Guerraoui, Rachid and Rodrigues, Lu{\'\i}s},
    year      = {2011},
    publisher = {Springer Science \& Business Media}
}

@article{baudet2019state,
    title   = {State machine replication in the libra blockchain},
    author  = {Baudet, Mathieu and Ching, Avery and Chursin, Andrey and Danezis, George and Garillot, Fran{\c{c}}ois and Li, Zekun and Malkhi, Dahlia and Naor, Oded and Perelman, Dmitri and Sonnino, Alberto},
    journal = {The Libra Assn., Tech. Rep},
    volume  = {7},
    year    = {2019}
}

@inproceedings{giridharan2024autobahn,
    title     = {Autobahn: Seamless high speed BFT},
    author    = {Giridharan, Neil and Suri-Payer, Florian and Abraham, Ittai and Alvisi, Lorenzo and Crooks, Natacha},
    booktitle = {Proceedings of the ACM SIGOPS 30th Symposium on Operating Systems Principles},
    pages     = {1--23},
    year      = {2024}
}

@inproceedings{danezis2022narwhal,
    title     = {Narwhal and tusk: a dag-based mempool and efficient bft consensus},
    author    = {Danezis, George and Kokoris-Kogias, Lefteris and Sonnino, Alberto and Spiegelman, Alexander},
    booktitle = {Proceedings of the Seventeenth European Conference on Computer Systems},
    pages     = {34--50},
    year      = {2022}
}

@inproceedings{spiegelman2025shoal,
    author    = {Alexander Spiegelman and Balaji Arun and Rati Gelashvili and Zekun Li},
    title     = {Shoal: Improving {DAG}-{BFT} Latency and Robustness},
    booktitle = {Financial Cryptography and Data Security (FC 2024), Revised Selected Papers, Part I},
    series    = {Lecture Notes in Computer Science},
    volume    = {14744},
    pages     = {92--109},
    publisher = {Springer},
    address   = {Cham},
    year      = {2025},
    doi       = {10.1007/978-3-031-78676-1_6},
    url       = {https://doi.org/10.1007/978-3-031-78676-1_6},
    note      = {FC 2024, Willemstad, Cura{\c c}ao, March 4--8, 2024}
}

@article{polyanskii2025starfish,
    title   = {Starfish: A high throughput BFT protocol on uncertified DAG with linear amortized communication complexity},
    author  = {Polyanskii, Nikita and Mueller, Sebastian and Vorobyev, Ilya},
    journal = {Cryptology ePrint Archive},
    year    = {2025}
}

@article{shrestha2024sailfish,
    title   = {Sailfish: Towards Improving the Latency of DAG-based BFT},
    author  = {Shrestha, Nibesh and Shrothrium, Rohan and Kate, Aniket and Nayak, Kartik},
    journal = {Cryptology ePrint Archive},
    year    = {2024}
}

@article{arun2024shoal++,
    title   = {Shoal++: High throughput dag bft can be fast!},
    author  = {Arun, Balaji and Li, Zekun and Suri-Payer, Florian and Das, Sourav and Spiegelman, Alexander},
    journal = {arXiv preprint arXiv:2405.20488},
    year    = {2024}
}

@inproceedings{spiegelman2022bullshark,
    title     = {Bullshark: Dag bft protocols made practical},
    author    = {Spiegelman, Alexander and Giridharan, Neil and Sonnino, Alberto and Kokoris-Kogias, Lefteris},
    booktitle = {Proceedings of the 2022 ACM SIGSAC Conference on Computer and Communications Security},
    pages     = {2705--2718},
    year      = {2022}
}

@inproceedings{babel2025mysticeti,
    title     = {MYSTICETI: Reaching the Latency Limits with Uncertified DAGs},
    author    = {Babel, Kushal and Chursin, Andrey and Danezis, George and Kichidis, Anastasios and Kokoris-Kogias, Lefteris and Koshy, Arun and Sonnino, Alberto and Tian, Mingwei},
    booktitle = {Network and Distributed Systems Security Symposium (NDSS)},
    year      = {2025}
}

@inproceedings{cohen2023proof,
    title        = {Proof of availability and retrieval in a modular blockchain architecture},
    author       = {Cohen, Shir and Goren, Guy and Kokoris-Kogias, Lefteris and Sonnino, Alberto and Spiegelman, Alexander},
    booktitle    = {International Conference on Financial Cryptography and Data Security},
    pages        = {36--53},
    year         = {2023},
    organization = {Springer}
}

@inproceedings{duan2024dashing,
    title     = {Dashing and Star: Byzantine fault tolerance with weak certificates},
    author    = {Duan, Sisi and Zhang, Haibin and Sui, Xiao and Huang, Baohan and Mu, Changchun and Di, Gang and Wang, Xiaoyun},
    booktitle = {Proceedings of the Nineteenth European Conference on Computer Systems},
    pages     = {250--264},
    year      = {2024}
}

@article{alhaddad2022asynchronous,
    title   = {Asynchronous verifiable information dispersal with near-optimal communication},
    author  = {Alhaddad, Nicolas and Das, Sourav and Duan, Sisi and Ren, Ling and Varia, Mayank and Xiang, Zhuolun and Zhang, Haibin},
    journal = {Cryptology ePrint Archive},
    year    = {2022}
}

@inproceedings{cachin2005asynchronous,
    title        = {Asynchronous verifiable information dispersal},
    author       = {Cachin, Christian and Tessaro, Stefano},
    booktitle    = {24th IEEE Symposium on Reliable Distributed Systems (SRDS'05)},
    pages        = {191--201},
    year         = {2005},
    organization = {IEEE}
}

@article{rabin1989efficient,
    title     = {Efficient dispersal of information for security, load balancing, and fault tolerance},
    author    = {Rabin, Michael O},
    journal   = {Journal of the ACM (JACM)},
    volume    = {36},
    number    = {2},
    pages     = {335--348},
    year      = {1989},
    publisher = {ACM New York, NY, USA}
}

@misc{sui,
    title        = {Sui},
    howpublished = {\url{https://sui.io/}},
    year         = {2024}
}

@inproceedings{shrestha2025optimistic,
    title     = {Optimistic, Signature-Free Reliable Broadcast and Its Applications},
    author    = {Shrestha, Nibesh and Yu, Qianyu and Kate, Aniket and Losa, Giuliano and Nayak, Kartik and Wang, Xuechao},
    booktitle = {Proceedings of the 2025 ACM SIGSAC Conference on Computer and Communications Security},
    year      = {2025}
}

@article{bracha1987asynchronous,
    title     = {Asynchronous Byzantine agreement protocols},
    author    = {Bracha, Gabriel},
    journal   = {Information and computation},
    volume    = {75},
    number    = {2},
    pages     = {130--143},
    year      = {1987},
    publisher = {Elsevier}
}

@inproceedings{abraham2021good,
    title     = {Good-case latency of byzantine broadcast: A complete categorization},
    author    = {Abraham, Ittai and Nayak, Kartik and Ren, Ling and Xiang, Zhuolun},
    booktitle = {Proceedings of the 2021 ACM Symposium on Principles of Distributed Computing},
    pages     = {331--341},
    year      = {2021}
}

@inproceedings{das2021asynchronous,
    title     = {Asynchronous data dissemination and its applications},
    author    = {Das, Sourav and Xiang, Zhuolun and Ren, Ling},
    booktitle = {Proceedings of the 2021 ACM SIGSAC Conference on Computer and Communications Security},
    pages     = {2705--2721},
    year      = {2021}
}

@misc{sailfish-code,
    title        = {Sailfish Codebase},
    author       = {{Nibesh Shrestha}},
    note         = {Accessed: 2025},
    year         = 2025,
    howpublished = {\url{https://github.com/nibeshrestha/sailfish}}
}

@misc{mysticeti-code,
    title        = {Mysticeti: Low-latency dag consensus with fast commit path},
    author       = {Mysten Labs},
    howpublished = {\url{https://github.com/asonnino/mysticeti}},
    year         = {2024}
}

@misc{rustcrypto-hashes,
    author       = {RustCrypto},
    title        = {RustCrypto: Hashes},
    howpublished = {\url{https://github.com/RustCrypto/hashes}},
    year         = 2024
}

@misc{ed25519-consensus,
    author       = {Henry de Valence},
    title        = {Ed25519 for consensus-critical contexts},
    howpublished = {\url{https://crates.io/crates/ed25519-consensus}},
    year         = 2024
}

@misc{tokio,
    author       = {The Tokio Team},
    title        = {Tokio},
    howpublished = {\url{https://tokio.rs}},
    year         = 2024
}

@misc{writev,
    author       = {Die.Net},
    title        = {writev(3) - Linux man page},
    howpublished = {\url{https://linux.die.net/man/3/writev}},
    year         = 2024
}

@misc{minibytes,
    author       = {Meta},
    title        = {Sapling (Minibytes)},
    howpublished = {\url{https://github.com/facebook/sapling/tree/main/eden/scm/lib/minibytes}},
    year         = 2024
}

@misc{sui-min-specs,
    author       = {The Sui Team},
    howpublished = {\url{https://docs.sui.io/guides/operator/validator/validator-config}},
    title        = {Validator Deployment amd Configuration},
    year         = {2025}
}

@inproceedings{twins,
  title={Twins: BFT Systems Made Robust},
  author={Bano, Shehar and Sonnino, Alberto and Chursin, Andrey and Perelman, Dmitri and Li, Zekun and Ching, Avery and Malkhi, Dahlia},
  booktitle={25th International Conference on Principles of Distributed Systems},
  year={2022}
}

@misc{hammerhead,
    title         = {HammerHead: Leader Reputation for Dynamic Scheduling},
    author        = {Giorgos Tsimos and Anastasios Kichidis and Alberto Sonnino and Lefteris Kokoris-Kogias},
    year          = {2023},
    eprint        = {2309.12713},
    archiveprefix = {arXiv},
    primaryclass  = {cs.CR},
    url           = {https://arxiv.org/abs/2309.12713}
}

@article{JACM:DwoLynSto,
    title     = {Consensus in the Presence of Partial Synchrony},
    author    = {Dwork, Cynthia and Lynch, Nancy and Stockmeyer, Larry},
    journal   = {Journal of the ACM},
    volume    = {35},
    number    = {2},
    pages     = {288--323},
    year      = {1988},
    publisher = {ACM},
    doi       = {10.1145/42282.42283}
}

@inproceedings{cohen2022aware,
    title        = {Be aware of your leaders},
    author       = {Cohen, Shir and Gelashvili, Rati and Kogias, Lefteris Kokoris and Li, Zekun and Malkhi, Dahlia and Sonnino, Alberto and Spiegelman, Alexander},
    booktitle    = {International Conference on Financial Cryptography and Data Security},
    pages        = {279--295},
    year         = {2022},
    organization = {Springer}
}

@phdthesis{buchman2016tendermint,
    title  = {Tendermint: Byzantine fault tolerance in the age of blockchains},
    author = {Buchman, Ethan},
    year   = {2016},
    school = {University of Guelph}
}

@inproceedings{castro1999practical,
    title     = {Practical byzantine fault tolerance},
    author    = {Castro, Miguel and Liskov, Barbara and others},
    booktitle = {OsDI},
    volume    = {99},
    number    = {1999},
    pages     = {173--186},
    year      = {1999}
}

@article{yang2019prism,
    title   = {Prism: Scaling bitcoin by 10,000 x},
    author  = {Yang, Lei and Bagaria, Vivek and Wang, Gerui and Alizadeh, Mohammad and Tse, David and Fanti, Giulia and Viswanath, Pramod},
    journal = {arXiv preprint arXiv:1909.11261},
    year    = {2019}
}

@misc{iota-consensus,
    titel        = {Consensus on IOTA},
    author       = {The IOTA team},
    howpublished = {\url{https://docs.iota.org/about-iota/iota-architecture/consensus}},
    year         = 2025
}

@misc{grass-route,
    title         = {Grassroots Systems: Concept, Examples, Implementation and Applications},
    author        = {Ehud Shapiro},
    year          = {2024},
    eprint        = {2301.04391},
    archiveprefix = {arXiv},
    primaryclass  = {cs.NI},
    url           = {https://arxiv.org/abs/2301.04391}
}

@misc{blocklace,
    title         = {The Blocklace: A Byzantine-repelling and Universal Conflict-free Replicated Data Type},
    author        = {Paulo Sérgio Almeida and Ehud Shapiro},
    year          = {2025},
    eprint        = {2402.08068},
    archiveprefix = {arXiv},
    primaryclass  = {cs.DC},
    url           = {https://arxiv.org/abs/2402.08068}
}

@article{green2021hashgraph,
    title     = {HashGraph—Scalable hash tables using a sparse graph data structure},
    author    = {Green, Oded},
    journal   = {ACM Transactions on Parallel Computing (TOPC)},
    volume    = {8},
    number    = {2},
    pages     = {1--17},
    year      = {2021},
    publisher = {ACM New York, NY, USA}
}

@misc{sui-code,
    title        = {Sui Blockchain},
    author       = {{MystenLabs}},
    note         = {Accessed: 2025},
    year         = 2025,
    howpublished = {\url{https://github.com/mystenlabs/sui}}
}

@misc{narwhal-code,
    title        = {Narwhal and Tusk Implementation},
    author       = {Alberto Sonnino},
    note         = {Accessed: 2025},
    year         = 2025,
    howpublished = {\url{https://github.com/asonnino/narwhal}}
}

@misc{Bullshark-code,
    title        = {Bullshark Implementation},
    author       = {Alberto Sonnino},
    note         = {Accessed: 2025},
    year         = 2025,
    howpublished = {\url{https://github.com/asonnino/narwhal/tree/bullshark}}
}

@misc{jolteon-code,
    title        = {Jolteon Implementation},
    author       = {Alberto Sonnino},
    note         = {Accessed: 2025},
    year         = 2025,
    howpublished = {\url{https://github.com/asonnino/hotstuff}}
}

@misc{hotstuff-code,
    title        = {HotStuff Implementation},
    author       = {Alberto Sonnino},
    note         = {Accessed: 2025},
    year         = 2025,
    howpublished = {\url{https://github.com/asonnino/hotstuff/tree/3-chain}}
}

@misc{diem-code,
    title        = {Diem Blockchain},
    author       = {{Diem Association}},
    note         = {Accessed: 2025},
    year         = 2025,
    howpublished = {\url{https://github.com/diem/diem}}
}

@misc{autobahn-code,
    title        = {Autobahn Artifact},
    author       = {Neil Giridharan},
    note         = {Accessed: 2025},
    year         = 2025,
    howpublished = {\url{https://github.com/neilgiri/autobahn-artifact}}
}

@misc{mahi-mahi-code,
    title        = {Mahi-Mahi Consensus Implementation},
    author       = {Pasindu Tennage},
    note         = {Accessed: 2025},
    year         = 2025,
    howpublished = {\url{https://github.com/PasinduTennage/mahi-mahi-consensus}}
}

@misc{ditto-code,
    title        = {Ditto Implementation},
    author       = {Daniel Xiang},
    note         = {Accessed: 2025},
    year         = 2025,
    howpublished = {\url{https://github.com/danielxiangzl/Ditto}}
}

@misc{cordial-miners,
    doi       = {10.4230/LIPICS.DISC.2023.26},
    url       = {https://drops.dagstuhl.de/entities/document/10.4230/LIPIcs.DISC.2023.26},
    author    = {Keidar, Idit and Naor, Oded and Poupko, Ouri and Shapiro, Ehud},
    language  = {en},
    title     = {Cordial Miners: Fast and Efficient Consensus for Every Eventuality},
    publisher = {Schloss Dagstuhl – Leibniz-Zentrum für Informatik},
    year      = {2023},
    copyright = {Creative Commons Attribution 4.0 International license}
}

@techreport{vyzovitis2020gossipsub,
    title       = {{GossipSub}: Attack-Resilient Message Propagation in the {Filecoin} and {ETH2.0} Networks},
    author      = {Vyzovitis, Dimitris and Napora, Yusef and McCormick, Dirk and Dias, David and Psaras, Yiannis},
    year        = {2020},
    institution = {Protocol Labs},
    number      = {PL-TechRep-2020-002},
    note        = {\url{https://arxiv.org/abs/2007.02754}}
}

@inproceedings{zhang2025no,
  title={No fish is too big for flash boys! frontrunning on DAG-based blockchains},
  author={Zhang, Jianting and Kate, Aniket},
  booktitle={2025 IEEE Annual Computer Security Applications Conference (ACSAC)},
  pages={1065--1080},
  year={2025},
  organization={IEEE}
}

@inproceedings{shrestha2025towards,
  title={Towards Improving Throughput and Scalability of DAG-based BFT SMR},
  author={Shrestha, Nibesh and Kate, Aniket},
  booktitle = {Proceedings of the 21st European Conference on Computer Systems},
  year      = {2026}
}

@inproceedings{Moura021,
  author       = {Leonardo {de Moura} and
                  Sebastian Ullrich},
  opteditor       = {Andr{\'{e}} Platzer and
                  Geoff Sutcliffe},
  title        = {{The Lean 4 Theorem Prover and Programming Language}},
  booktitle    = {CADE},
  series       = {LNCS},
  pages        = {625--635},
  publisher    = {Springer},
  year         = {2021},
  url          = {https://doi.org/10.1007/978-3-030-79876-5\_37},
  doi          = {10.1007/978-3-030-79876-5\_37},
}

@article{kichidis2025beluga,
  title={Beluga: Block Synchronization for BFT Consensus Protocols},
  author={Kichidis, Tasos and Kokoris-Kogias, Lefteris and Koshy, Arun and Sergey, Ilya and Sonnino, Alberto and Tian, Mingwei and Zhang, Jianting},
  journal={arXiv preprint arXiv:2511.15517},
  year={2025}
}

\ifarxiv
\appendix
\section{Open Science}

All artifacts necessary for evaluating the contributions of this paper are publicly available at:
\begin{center}
  \codelink
\end{center}
\Cref{sec:tutorial} additionally provides a tutorial to reproduce experiments.

\subsection{Reproducing Experiments} \label{sec:tutorial}

We provide the orchestration scripts\footnote{\codelink} used to benchmark the \sysname codebase on AWS and produce the benchmarks of \autoref{sec-experiment}.

\para{Deploying a testbed}
The file `\texttildelow/.aws/credentials' should have the following content:
\begingroup
\scriptsize
\begin{verbatim}
[default]
aws_access_key_id = YOUR_ACCESS_KEY_ID
aws_secret_access_key = YOUR_SECRET_ACCESS_KEY
\end{verbatim}
\endgroup
\noindent configured with account-specific AWS \emph{access key id} and \emph{secret access key}. It is advise to not specify any AWS region as the orchestration scripts need to handle multiple regions programmatically.

A file `settings.yaml' contains all the configuration parameters for the testbed deployment. We run the experiments of \autoref{sec-experiment} with the following settings:

\begingroup
\scriptsize\begin{lstlisting}[language=yaml]
---
testbed_id: "${USER}-testbed"
cloud_provider: aws
token_file: "/Users/${USER}/.aws/credentials"
ssh_private_key_file: "/Users/${USER}/.ssh/aws"
regions:
  - us-east-1
  - us-west-2
  - ca-central-1
  - eu-central-1
  - eu-west-1
  - eu-west-2
  - eu-west-3
  - eu-north-1
  - ap-south-1
  - ap-southeast-1
  - ap-southeast-2
  - ap-northeast-1
  - ap-northeast-2
specs: m5d.8xlarge
repository:
  url: https://github.com/AUTHOR/REPO.git
  commit: main
node_parameters_path:
  "crates/orchestrator/assets/node-parameters.yml"
client_parameters_path:
  "crates/orchestrator/assets/client-parameters.yml"
benchmark_duration: 1000
\end{lstlisting}
\endgroup

where the file `/Users/\${USER}/.ssh/aws' holds the ssh private key used to access the AWS instances, and `AUTHOR' and `REPO' are respectively the GitHub username and repository name of the codebase to benchmark.

The orchestrator binary provides various functionalities for creating, starting, stopping, and destroying instances. For instance, the following command to boots 2 instances per region (if the settings file specifies 13 regions, as shown in the example above, a total of 26 instances will be created):

\begingroup
\scriptsize\begin{verbatim}
cargo run --bin orchestrator -- testbed deploy --instances 2
\end{verbatim}
\endgroup
The following command displays he current status of the testbed instances
\begingroup
\scriptsize\begin{verbatim}
cargo run --bin orchestrator testbed status
\end{verbatim}
\endgroup
Instances listed with a green number are available and ready for use and instances listed with a red number are stopped. It is necessary to boot at least one instance per load generator, one instance per validator, and one additional instance for monitoring purposes (see below).
The following commands respectively start and stop instances:
\begingroup
\scriptsize\begin{verbatim}
cargo run --bin orchestrator -- testbed start
cargo run --bin orchestrator -- testbed stop
\end{verbatim}
\endgroup
It is advised to always stop machines when unused to avoid incurring in unnecessary costs.

\para{Running Benchmarks}
Running benchmarks involves installing the specified version of the codebase on all remote machines and running one validator and one load generator per instance. For example, the following command benchmarks a committee of 50 validators under a constant load of 1,000 tx/s:
\begingroup
\scriptsize\begin{verbatim}
cargo run --bin orchestrator -- benchmark \
    --committee 50 fixed-load --loads 1000
\end{verbatim}
\endgroup
The nodes and clients configuration files (respectively specified in `crates/orchestrator/assets/node-parameters.yml' and \\
`crates/orchestrator/assets/client-parameters.yml') are used to respectively parametrize the nodes and clients. We run our benchmarks with the nodes configuration file as follows:
{\scriptsize\begin{lstlisting}[language=yaml]
leader_timeout:
  secs: 1
  nanos: 0
\end{lstlisting}}
and the client configuration file as follows:
{\scriptsize\begin{lstlisting}[language=yaml]
initial_delay:
  secs: 400
  nanos: 0
\end{lstlisting}}

\para{Monitoring}
The orchestrator provides facilities to monitor metrics. It deploys a Prometheus instance and a Grafana instance on a dedicated remote machine. Grafana is then available on the address printed on \texttt{stdout} when running benchmarks with the default username and password both set to \texttt{admin}. An example Grafana dashboard can be found in the file `grafana-dashboard.json'\footnote{\dashboardlink}.

\section{Ethical Considerations} \label{sec-ethical-considerations}
This work presents a pull-induction attack and shows how it can be launched to degrade the performance of many modern BFT consensus protocols. We conduct our attack evaluation in a self-deployed environment in a controlled way, which does not affect any production systems. While the present attacks might be exploited in real-world systems that have adopted the corresponding protocols, the primary contribution is a mitigation. We have disclosed this vulnerability to the relevant Sui team and helped them integrate our mitigation into their system.
\section{Performance Analysis} \label{sec-performance-analysis}
We prove the push latency of \sysname both in the happy case and under adversarial settings.

\subsection{Performance Analysis in the Happy Case} \label{sec-performance-happy-case}
In this section, we will show that \sysname can achieve optimal push latency (i.e., $\Delta$) under happy cases and the push latency of nearly $2\Delta$ under adverse cases after GST.



In happy cases, all responsive validators (at least $2f{+}1$) are honest and share their blocks in time. According to \Cref{lem-round-liveness}, all honest validators can enter the same round (w.l.o.g. at round $r$) within $4\Delta$ after GST and can create their round $r$ blocks by time $GST{+}4\Delta$. As a result, every honest validator must be able to receive a quorum $\mathcal{B}^r$ containing at least $2f{+}1$ round $r$ blocks by time $GST{+}5\Delta$. Since all responsive validators are honest in happy cases, every honest validator will receive $\mathcal{B}^r$ at time, w.l.o.g., $t_{sync} {<} GST{+}5\Delta$ and move to the next round $r{+}1$ at $t_{sync}$. After that, all honest validators can receive at least $2f{+}1$ round $r{+}1$ blocks by time $t_{sync}{+}\Delta$ and move to round $r{+}2$ at time $t_{sync}{+}\Delta$. By induction on rounds, we can see that the push latency under happy cases is $\Delta$.

\subsection{Performance Analysis under Adverse cases} \label{appendix-full-performance-analysis-adverse}
In this section, we give a rigorous proof to show that \sysname can achieve a push latency of nearly $2\Delta$ under adverse cases.
The proof relies on the following assumption. (Note that the correctness of \sysname (cf. \autoref{sec-correctness-analysis}) does not rely on these assumptions.)

\begin{assumption}[Latency Triangle] \label{asm-latency-triangle}
    After GST, the direct network latency between any pair of honest validators is always faster than going through an intermediate validator.
\end{assumption}


In adverse cases, there are at most $f$ malicious validators that aim to increase the push latency by inducing honest validators to trigger the pull protocol. In the following, we denote the set of honest validators as $\mathcal{V}_h$ and the set of malicious validators as $\mathcal{A}$.

\begin{lemma} \label{lem-no-blame-honest}
    After GST, each honest validator will not get blamed and have its reputation decreased by honest validators.
\end{lemma}
\begin{proof}
    Recall from \autoref{sec-sysname-push} that a validator $v_i$ has its reputation decreased by honest validators only if $f{+}1$ validators report having invoked the pull protocol to synchronize $v_i$'s blocks. An honest validator $v_j$ invokes the pull protocol upon receiving a block $B_k^r$ from a validator $v_k$ that references $v_i$'s round $r{-}1$ block $B_i^{r{-}1}$, while $v_j$ has not yet received $B_i^{r{-}1}$ directly. By \Cref{asm-latency-triangle}, however, after GST, if $v_i$ is honest and sends $B_i^{r{-}1}$ to $v_j$, then $v_j$ receives $B_i^{r{-}1}$ directly from $v_i$ before receiving it through any intermediate validator $v_k$ that references it in $B_k^r$. Hence, $v_j$ does not invoke the pull protocol for $B_i^{r{-}1}$ and does not report $v_i$. Since there are at most $f$ malicious validators, an honest $v_i$ cannot be reported by $f{+}1$ validators and therefore cannot be blamed by honest validators.
\end{proof}

\begin{lemma} \label{lem-optimal-round-latency-condition}
    After GST, if all honest validators enter round $r$ at time $t_r$ and have their reputation higher than that of any malicious validator, then for any future round $r'\geq r$, the latency of round $r'$ is $\Delta$.
\end{lemma}
\begin{proof}
    Since all honest validators have higher reputation than any malicious validator, by the reputation-based push protocol, honest validators reference only round $r{-}1$ blocks from honest validators when creating their round $r$ blocks. After GST, these round $r$ blocks are accepted without invoking the pull protocol. Consequently, round $r$ has latency $\Delta$, and all honest validators enter round $r{+}1$ at time $t_r{+}\Delta$. By induction, the latency of any future round $r'\geq r$ is $\Delta$.
\end{proof}

\begin{lemma} \label{lem-bounded-latency-each-round}
    After GST, if all honest validators enter round $r$ at time $t_r$, then either the expected latency of any future round $r'{>}r$ is within $2\Delta$ or at least one malicious validator is blamed by honest validators, in which case the latency of round $r'$ is at most $4\Delta$.
\end{lemma}
\begin{proof}
    Recall that if an honest validator $v_i$ enters round $r$ at time $t_r$, then $v_i$ must have received at least $2f{+}1$ acceptable round $r{-}1$ blocks and can create its round $r$ block at $t_r$.
    Since all honest validators enter round $r$ at time $t_r$, every honest validator will receive at least $2f{+}1$ round $r$ blocks created by honest validators by time $t_r{+}\Delta$. There are three cases.

    \textit{Case 1:} If all honest validators have their reputation higher than that of any malicious validator, then according to \Cref{lem-optimal-round-latency-condition}, the latency of future round $r'{>}r$ is $\Delta$.

    \textit{Case 2:} If it is not the case 1, and for each round $r{-}1$ block created by malicious validators $\mathcal{A}$, it is referenced by at least $f{+}1$ round $r$ blocks created by honest validators, i.e., $\mathcal{A}$ share their round $r{-}1$ blocks with at least $f{+}1$ honest validator. In this case, thanks to the ImPoA-based pull mechanism, every honest validator can accept all round $r$ blocks created by honest validators without synchronizing any missing blocks on the push path and create its round $r{+}1$ block at time $t_r{+}\Delta$. The latency of round $r{+}1$ is $\Delta$. By induction, we can see that for any future round $r'\geq r$, the latency of round $r'$ is $\Delta$.

    \textit{Case 3:} If it is not the case 1, and at least one round $r{-}1$ block $B_m^{r-1}$ created by a malicious validator $v_m$ is referenced by fewer than $f{+}1$ round $r$ blocks created by honest validators, i.e., $v_m$ delays or did not share $B_m^{r-1}$ with more than $f$ honest validators. We denote those honest validators referencing $B_m^{r-1}$ by $t_r$ as $\mathcal{V}_h^{Ref}$, and those who do not as $\mathcal{V}_h^{NoR}$. There are two scenarios, and we show that for any future round $r'\geq r{+}1$, the expected latency of $r'$ is at most $2\Delta$ before $v_m$ is blamed.

    First, $v_m$ will never share $B_m^{r-1}$ with $\mathcal{V}_h^{NoR}$, then $|\mathcal{V}_h^{NoR}|\geq f{+}1$ honest validators will report $v_m$, and $v_m$ will be blamed by honest validators.

    Second, $v_m$ delays sharing $B_m^{r-1}$ with some honest validators $\mathcal{V}_{h1}^{NoR}$ but not the others $\mathcal{V}_{h2}^{NoR}$ to escape being blamed, where $|\mathcal{V}_{h2}^{NoR}=\mathcal{V}_h^{NoR}\setminus \mathcal{V}_{h1}^{NoR}|\leq f$. In this scenario, note that $\mathcal{V}_{h1}^{NoR}$ must have received $B_m^{r-1}$ by time $t_r{+}\Delta$, since otherwise, $\mathcal{V}_{h1}^{NoR}$ learn $B_m^{r-1}$ is missing from $\mathcal{V}_h^{Ref}$'s round $r$ blocks and will report $v_m$. Thus, both $\mathcal{V}_h^{Ref}$ and $\mathcal{V}_{h1}^{NoR}$ can create their round $r{+}1$ blocks by time $t_r{+}\Delta$. The delayed honest validators $\mathcal{V}_{h2}^{NoR}$, instead, must pull $B_m^{r-1}$ at time $t_r{+}\Delta$ (when they receive round $r$ blocks from $\mathcal{V}_h^{Ref}\cup \mathcal{V}_{h1}^{NoR}$) and pull any missing round $r$ blocks at time $t_r{+}2\Delta$ (when they receive round $r{+}1$ blocks from $\mathcal{V}_h^{Ref}\cup \mathcal{V}_{h1}^{NoR}$). As a result, even for the delayed honest validators $\mathcal{V}_{h2}^{NoR}$, they can create their round $r{+}1$ blocks by time $t_r{+}3\Delta$ and create their round $r{+}2$ blocks $t_r{+}4\Delta$. Since the non-delayed $\mathcal{V}_h^{Ref}\cup \mathcal{V}_{h1}^{NoR}$ can create their round $r{+}2$ blocks by time $t_r{+}4\Delta$, such a delaying process can be repeated every two rounds. As a result, the maximum average latency of each future round $r'\geq r{+}1$ is delayed by at most $2\Delta$, if $v_m$ wishes to escape being blamed.

    In addition, when $v_m$ is blamed due to it delaying a round $r''$, according to \Cref{lem-round-liveness}, the latency of round $r''$ is at most $4\Delta$. The proof is done.

\end{proof}

\begin{lemma} \label{lem-bounded-delay}
    After GST, once all honest validators enter round $r$, for future round $r'{>}r$ that malicious validators $\mathcal{A}$ delay, its expected latency is within $2\Delta$, or at least one malicious validator is blamed by honest validators.
\end{lemma}
\begin{proof}
    Assume round $r$ is the highest round that all honest validators are at some time $t$. By \Cref{lem-round-liveness}
    all honest validator will enter round $r$ by $t_r= t{+}4\Delta$ after GST. 

    Consider two sets of honest validators: slow honest validators $\mathcal{V}_h^{slw}$ and fast honest validators $\mathcal{V}_h^{fst}=\mathcal{V}_h\setminus \mathcal{V}_h^{slw}$. Validators in $\mathcal{V}_h^{slw}$ need to invoke the pull protocol for missing round $r{-}1$ blocks to enter round $r$ by time $t_r$, while validators in $\mathcal{V}_h^{fst}$ do not. Apparently, we have $|\mathcal{V}_h^{fst}| \geq f{+}1$, since otherwise, there are more than $f{+}1$ honest validators in $\mathcal{V}_h^{slw}$ reporting $\mathcal{A}$, and $\mathcal{A}$ would be blamed. 
    Since all honest validators can enter round $r$ by $t_r$, each honest validator will receive round $r$ blocks from all honest validators by $t_r+\Delta$. Then there are two cases:

    \textit{Case 1:} If all honest validators enter round $r{+}1$ at time $t_r{+}\Delta$, according to \Cref{lem-bounded-latency-each-round}, either the expected latency of round $r{+}2$ is within $2\Delta$ or at least one malicious validator is blamed by honest validators.

    \textit{Case 2:} Otherwise, honest validators in $\mathcal{V}_h^{fst}$ enter round $r{+}1$ before they receive $\mathcal{V}_h^{slw}$'s round $r$ blocks at $t_r+\Delta$. There are two scenarios, and we show that the expected latency of round $r{+}2$ is within $2\Delta$ before at least one malicious validator is blamed.

    First, every round $r$ block created by malicious validators $\mathcal{A}$ is referenced by at least $f{+}1$ round $r{+}1$ blocks created by honest validators in $\mathcal{V}_h^{fst}$. In this scenario, thanks to the ImPoA-based pull mechanism, every honest validator in $\mathcal{V}_h^{slw}$ can accept these round $r{+}1$ blocks without synchronizing any missing blocks on the push path. As a result, $\mathcal{V}_h^{slw}$ can create their round $r{+}2$ blocks at time $t_r+2\Delta$. Since, $\mathcal{V}_h^{slw}$ create their round $r{+}1$ blocks at time $t_r+\Delta$ (as mentioned above), the latency of round $r{+}2$ is within $\Delta$.

    Second, at least one round $r$ block created by malicious validators $\mathcal{A}$ is referenced by fewer than $f{+}1$ round $r{+}1$ blocks created by honest validators in $\mathcal{V}_h^{fst}$. This means that at least one malicious validator $v_m$ delays sharing its round $r$ block $B_m^r$ with some validators in $\mathcal{V}_{h1}^{fst}\subset \mathcal{V}_h^{fst}$, because otherwise, all honest validators in $\mathcal{V}_h^{fst}$ can receive enough round $r$ blocks and enter round $r{+}1$ at the same time. However, all honest validators in $\mathcal{V}_{h1}^{fst}$ must be able to create round $r{+}2$ blocks by time $t_r{+}2\Delta$, since the delayed validators in $\mathcal{V}_h^{slw}$ create their round $r{+}1$ blocks at time $t_r+\Delta$ (as mentioned above). As a result, once receiving $\mathcal{V}_{h1}^{fst}$'s round $r{+}2$ blocks at time $t_r{+}3\Delta$, $\mathcal{V}_h^{slw}$ can invoke the pull protocol to fetch any missing round $r{+}1$ blocks within $2\Delta$, and create their round $r{+}3$ blocks by time $t_r{+}5\Delta$. Recall that $\mathcal{V}_h^{slw}$ create their round $r{+}1$ blocks at time $t_r+\Delta$. For $\mathcal{V}_h^{slw}$, there are two rounds $r{+}1$ and $r{+}2$ between time $t_r{+}\Delta$ and $t_r{+}5\Delta$. As a result, the average latency of these two rounds is $(t_r{+}5\Delta-t_r{+}\Delta)/2=2\Delta$.


    By induction, we can see that for any future round $r'{>} r$, the latency of round $r'$ is within $2\Delta$ before any malicious validator $v_m$ is blamed. The proof is done.

\end{proof}

\begin{lemma} \label{lem-bounded-number{-}of-delayed-rounds}
    After GST, malicious validators $\mathcal{A}$ in \sysname can only delay the progress of the protocol in a bounded number of rounds with an expected latency higher than $2\Delta$ every $R_L$ rounds.
\end{lemma}
\begin{proof}
    Recall that, when blamed, a validator has its reputation decreased by $R_L$ by every honest validator, where $R_L \geq 1$ is a predetermined parameter. Without loss of generality, assume that right after GST each malicious validator has reputation $R_m$, and the lowest reputation of any honest validator is $R_h$.

    By \Cref{lem-bounded-delay}, $\mathcal{A}$ can delay the protocol for one round with latency at most $5\Delta$ without being blamed. Denote this period by $\mathcal{D}_1$, and let $|\mathcal{D}_1|$ denote the number of rounds it delays with expected latency exceeding $2\Delta$. Thus, $\mathcal{D}_1$ adds at most $4\Delta$ to the total latency before $\mathcal{A}$ is blamed.

    After period $\mathcal{D}_1$, all honest validators can enter the same round at the same time. According to \Cref{lem-bounded-latency-each-round}, $\mathcal{A}$ might still be able to delay some future rounds with the expected latency of more than $2\Delta$, but each of such round will lead to at least one malicious validator being blamed and getting lose of $R_L$ points. Based on the reputation difference, we can derive that $f$ malicious validators can delay the protocol for $(R_m-R_h)\cdot f/R_L$ rounds before their reputations are lower than that of any honest validator. We denote this period as $\mathcal{D}_2$, and similarly, we have $|\mathcal{D}_2|=(R_m-R_h)\cdot f/R_L$. According to \Cref{lem-bounded-latency-each-round}, each of these rounds will increase the push latency by at most $3\Delta$. As a result, $\mathcal{D}_2$ will increase the latency of the protocol by at most $(R_m-R_h)\cdot f/R_L\cdot 3\Delta$.

    After period $\mathcal{D}_2$, $\mathcal{A}$ have the reputation equal to the lowest reputation of an honest validator. Since according to \Cref{lem-no-blame-honest}, the honest validators will not get their reputations decreased after GST, $\mathcal{A}$ need to perform correctly to get their reputation increased by honest validators. In particular, to delay $f$ rounds with the expected latency of more than $2\Delta$ for each after $\mathcal{D}_2$, these $f$ malicious validators $\mathcal{A}$ perform carefully without getting blamed for at least $R_L$ rounds, during which the expected latency of each round will be $2\Delta$. We denote this period as $\mathcal{D}_3$. During $\mathcal{D}_3$, $\mathcal{A}$ delay $f$ rounds with the expected latency of higher than $2\Delta$ every $R_L$ rounds. Thus, we have $|\mathcal{D}_3|=f$. According to \Cref{lem-bounded-latency-each-round}, each round being delayed will increase the push latency by at most $2\Delta$. As a result, $\mathcal{D}_3$ will increase the latency of the protocol by at most $f\cdot 2\Delta$ every $R_L$ rounds.

    By considering the above three periods, we can conclude that $\mathcal{A}$ can delay the protocol with the expected push latency higher than $2\Delta$ for at most $|\mathcal{D}_1|{+}|\mathcal{D}_2|{+}|\mathcal{D}_3|=1{+}(R_m-R_h)\cdot f/R_L{+}f$ rounds, and the extra latency introduced by $\mathcal{A}$ is at most $(4\Delta+(R_m-R_h)\cdot f/R_L\cdot 3\Delta{+}f\cdot 2\Delta)$ every $R_L$ rounds.
    As $f, R_m, R_h$, and $R_L$ are all constants, the proof is done.

\end{proof}

Finally, we have a proof for \Cref{thm-optimal-round-latency}.
\begin{theorem} \label{thm-optimal-round-latency}
    After GST, \sysname can achieve a push latency of nearly $2\Delta$ under adverse cases.
\end{theorem}
\begin{proof}[Proof for \Cref{thm-optimal-round-latency}]
    According to \Cref{lem-bounded-number{-}of-delayed-rounds}, for every $R_L$ rounds, malicious validators $\mathcal{A}$ can only increase the push latency by more than $2\Delta$ for a bounded number of rounds (corresponding to the $\mathcal{D}_1$, $\mathcal{D}_2$, and $\mathcal{D}_3$ periods). For the remaining rounds, the expected push latency is $2\Delta$. Hence, the total latency over $R_L$ rounds can be upper-bounded by $2\Delta + \Delta_{\text{extra}}$, where $\Delta_{\text{extra}}$ denotes the additional delay introduced by the adversarial periods.

    From \Cref{lem-bounded-number{-}of-delayed-rounds}, this extra delay satisfies $\Delta_{\text{extra}}\le 4\Delta+\frac{(R_m - R_h)f}{R_L}\cdot 3\Delta{+}f\cdot 2\Delta$. Therefore, the average push latency over these $R_L$ rounds is at most $\frac{R_L \cdot  2\Delta + \Delta_{\text{extra}}}{R_L}= 2\Delta \;+\; \frac{\Delta_{\text{extra}}}{R_L}$. Substituting the bound on $\Delta_{\text{extra}}$, we obtain $2\Delta \;+\; \frac{4\Delta}{R_L}\;+\; \frac{(R_m - R_h)f}{R_L^2} \cdot  3\Delta\;+\; \frac{f}{R_L} \cdot  2\Delta$.

    Observe that the delays caused by the $\mathcal{D}_1$ (corresponding to $\frac{4\Delta}{R_L}$ term) and $\mathcal{D}_2$ (corresponding to $\frac{(R_m - R_h)f}{R_L^2} \cdot 3\Delta$ term) periods occur only once after GST. Thus, the latencies caused by $\mathcal{D}_1$ and $\mathcal{D}_2$ become negligible over a sufficiently long execution and when $R_L$ is large. As a result, asymptotically, the average push latency is $2\Delta(1+\frac{f}{R_L})=2\kappa\Delta$ with $\kappa=1+\frac{f}{R_L}$. By setting a sufficiently large $R_L$, the average push latency can be arbitrarily close to $2\Delta$.
\end{proof}
\section{Security Analysis of Mysticeti-\sysname} \label{sec-sec-analysis}
\autoref{sec-sysname-consensus} discusses how to build a BFT consensus on \sysname to achieve a more efficient and robust block synchronization, and in \autoref{sec-experiment}, we instantiate Mysticeti consensus over \sysname to obtain a variant consensus Mysticeti-\sysname, demonstrating the performance improvement in practice. This section gives the security proof of Mysticeti-\sysname.

Like most BFT consensus protocols, we consider two security properties for Mysticeti-\sysname: \emph{liveness} and \emph{safety}. Similar to Mysticeti, the security properties of Mysticeti-\sysname only rely on strong links.\footnote{Mysticeti only defines strong links (or parents) but not weak links.} For simplicity, all references and links used in this section are considered strong links if not stated otherwise.
To help readers better follow the proof and make the section self-contained, we first review the consensus logic of Mysticeti and show how to integrate \sysname into Mysticeti (\autoref{sec-mysticeti-sysname-review}), and then prove the liveness (\Cref{sec-my-sys-liveness}) and safety (\Cref{sec-my-sys-safety}) of Mysticeti-\sysname.

\subsection{Build Mysticeti on \sysname} \label{sec-mysticeti-sysname-review}
\subsubsection{Consensus Logic in Mysticeti} \label{sec-consensus-logic-mysticeti}
Mysticeti orders transactions on top of the DAG constructed by its block synchroniser module using a two-step scheme. Below, we give the necessary definitions and notations required by the security proof, but recommend readers refer to Mysticeti's paper for more details about its consensus logic.

\para{(1) Decide status of leaders.}
First, Mysticeti designates leader validators using a round-robin approach\footnote{We consider single-leader when building Mysticeti-\sysname, where each round contains only one leader.} and employs two decision rules to decide the statuses of leader blocks that are proposed by the designated leader validators.
To achieve this, Mysticeti defines two DAG patterns on a leader block $B_L$: (i) \textit{skip pattern}, if $B_L$ is \textit{not} referenced by $2f{+}1$ blocks of the successive round; (ii) \textit{certificate pattern},\footnote{The same pattern as we defined in \autoref{sec-sysname-consensus} of this paper.} if $B_L$ is referenced by at least $2f{+}1$ blocks of the next round. Any subsequent block that includes a certificate pattern on $B_L$ via strong links in its causal history is called a \textit{certificate} for $B_L$ and is considered to certify $B_L$. For instance, in \Cref{fig:impoa}(a), the round $r{-}1$ block $B_2^{r-1}$ is referenced by $\{B_1^r, B_2^r, B_4^r\}$ from round $r$, and thus forms a certificate pattern. Since $B_2^{r+1}$ from round $r{+}1$ references $\{B_1^r, B_2^r, B_4^r\}$, $B_2^{r+1}$ is a certificate for $B_2^{r-1}$.

Validators then employ the following decision rules to decide leader blocks as \textit{to-commit} status or \textit{to-skip} status:
{
    \makeatletter
    \def\@listi{\leftmargin10pt \labelwidth\z@ \labelsep5pt}
    \makeatother
    \begin{itemize}
        \item \textit{Direct decision rule:} A round $r$ leader block $B_L^r$ is directly decided if it (i) has at least $2f{+}1$ certificate patterns at round $r{+}2$, in which $B_L^r$ is decided as \textit{to-commit}, or (ii) is identified as a \textit{skip pattern}, in which $B_L^r$ is decided as \textit{to-skip}. Otherwise, it remains undecided.
        \item \textit{Indirect decision rule:} For any undecided round $r'$ leader block $B_L^{r'}$, a replica searches for the first subsequent leader block $B_L^{r''}$ (where $r''{>}r'{+}2$) that is either decided as to-commit or is still undecided. In the former case, if $B_L^{r''}$ causally references a certificate for $B_L^{r'}$, then $B_L^{r'}$ is decided as to-commit; otherwise, it is decided as to-skip. In the later case, $B_L^{r'}$ remains undecided for the moment (Mysticeti guarantees that it will eventually be decided).
    \end{itemize}
}

If a leader block is decided as \textit{to-commit}, we call that the validators commit the leader block. Similarly, if a leader block is decided as \textit{to-skip}, we call that the validators skip the leader block.

\para{(2) Order transactions.}
Second, after all leader blocks up to round $r$ are decided, validators order all leader blocks up to round $r$ that have been decided as \textit{to-commit} and their causal history blocks using a deterministic linearization algorithm, e.g., a breadth-first search traversal~\cite{zhang2025no}. Recall that in the Mysticeti's original block synchronizer module (i.e., \syncopt synchronizer protocol), validators reference a leader block only when they receive \emph{all transactions} of the leader block's causal history. As a result, once blocks are ordered, transactions in the blocks are ordered and can be finalized as well. Here, the transaction finalization is defined based on the specific protocols/applications. For instance, in the state machine replication (SMR), transaction finalization indicates that the transaction is executed. Nevertheless, they all require that validators have full transaction data.

\para{Leader timeout.}
Since Mysticeti operates in the partial synchrony network model, it relies on timeouts to ensure liveness after GST. Mysticeti's consensus logic introduces a leader timeout $T_{ld}$ to round advancement. In particular, validators \emph{is ready to} advance to round $r$ (i) upon receiving the round $r{-}1$ leader block and receiving $2f{+}1$ round $r{-}1$ blocks referencing the round $r{-}2$ leader block, or (ii) upon the expiration of $T_{ld}$.

\subsubsection{Integrate \sysname into Mysticeti} \label{sec-integrate-sysname-to-mysticeti}
By replacing Mysticeti's original synchronizer module with \sysname, we can derive a variant consensus protocol Mysticeti-\sysname. Below, we demonstrate how \sysname interfaces with Mysticeti's consensus logic (\Cref{fig:overview} presents the overview of the interfaces). It is worth noting that Mysticeti-\sysname does not modify Mysticeti's consensus logic.

\para{Round advancement.}
In Mysticeti-\sysname, a validator $v_i$ advances to round $r$ and calls $\dagbc_i(B,r)$ to create and disseminate a round $r$ block $B$ to the system. Note that both \sysname and Mysticeti's consensus logic define rules for round advancement. Mysticeti-\sysname incorporates all these rules to preserve the relevant properties. Specifically, $v_i$ advances to round $r$ and calls $\dagbc_i(B,r)$ with at least $2f{+}1$ round $r{-}1$ accepted blocks and if one of the following round advancement rules becomes satisfied: (i) $v_i$ receives $2f{+}1$ blocks from round $r{-}1$ whose creators have reputations above the threshold $R_t$ \emph{and} receives the round $r{-}1$ leader block $B_L^{r-1}$ \emph{and} receives $2f{+}1$ round $r{-}1$ blocks referencing $B_L^{r-2}$; or (ii) $v_i$ is in round $r{-}1$, and a timeout $T_{live}=\max\{T_{ld}, T_{rd}\}$ expires, where $T_{ld}$ is the leader timeout introduced by Mysticeti's consensus logic, and $T_{rd}$ is the pre-round timeout used by \sysname for its round advancements in the synchronizer module; or (iii) $v_i$ is in round ${<}r$ and observes some blocks of round ${>}r$. In Mysticeti-\sysname, $T_{ld}$ is set to $7\Delta$, and both $T_{ld}$ and $T_{rd}$ are per-round timeouts (i.e., $v_i$ reschedules them once advancing to a new round). As a result, we use $T_{live}=7\Delta$ solely in our liveness proof for illustration purposes.

\para{Parent selection.}
When creating a new block, $v_i$ refers to its admission control rules to prioritize selecting high-reputation blocks that have been output via $\dagpoa$ and are considered \emph{consensus-specified}. A block in round $r{-}1$ is considered consensus-specified if it references the round $r{-}2$ leader block or itself is the round $r{-}1$ leader block. Note that including consensus-specified blocks as parents is required by Mysticeti's consensus logic (for liveness) and is not inherently required by \sysname. Nevertheless, \sysname can help validators filter out consensus-specified blocks that are created by suspected Byzantine validators and prevent them from delaying the progress of pushing blocks.


\para{Block ordering and transaction finalization.}
When ordering blocks, $v_i$ executes Mysticeti's decision rules on leader blocks that have been output via $\dagpoa$. In particular, similar to the vanilla Mysticeti, Mysticeti-\sysname employs a two-step scheme to order blocks. First, validators decide the status of every leader block using the two underlying decision rules. Then, validators order all leader blocks that have been decided as \textit{to-commit} status and their causal history blocks. Transactions in a block $B$ in Mysticeti-\sysname are finalized if \emph{all blocks} that are ordered before $B$ have been output via $\dagdeli$ and their transactions have been finalized.

\subsection{Liveness of Mysticeti-\sysname} \label{sec-my-sys-liveness}
We first show that Mysticeti-\sysname ensures liveness under partial synchrony. The liveness states that transactions will eventually be ordered and finalized by validators.


\begin{lemma} \label{lem-leader-proposal}
    In Mysticeti-\sysname, after GST, an honest validator's leader block will be referenced in the next round by every honest validator.
\end{lemma}
\begin{proof}
    After GST, if an honest validator enters a round $r$ at time $t$ and $r$ is the highest round among honest validators, then the honest leader validator (and every other honest validator) will be able to enter the same round $r$ by time $t+4\Delta$ (\Cref{lem-round-liveness}). Then the honest leader validator will directly create and disseminate the round $r$ leader block $B_L^r$, which will take another $\Delta$ to be received by every honest validator. Thus, by time $t+5\Delta$, all honest validators will receive the $B_L^r$, and with our pull protocol, they will accept $B_L^r$ by time $t+7\Delta$. Since the timeout $T_{live}$ is set to $7\Delta$, $B_L^r$ will arrive and become acceptable before the first honest validator times out. Since validators are asked to include leader blocks as parents (see \Cref{sec-integrate-sysname-to-mysticeti}), every honest validator will vote for the leader block.
\end{proof}

\begin{lemma} \label{lem-sufficient-votes}
    In Mysticeti-\sysname, after GST, all honest validators will create a certificate for the leader block proposed by an honest validator.
\end{lemma}
\begin{proof}
    Assume there is an honest leader block $B_L^{r}$ in round $r$, by \Cref{lem-leader-proposal}, all honest validators will vote for $B_L^{r}$ after GST. This means that $B_L^{r}$ is a certified block, and all honest validators will have their round $r{+}1$ blocks referencing $B_L^{r}$ as parents. By \Cref{lem-round-liveness}, every honest validator can receive $2f{+}1$ honest blocks from round $r{+}1$ within $5\Delta$, and these blocks will become acceptable within $7\Delta$. Consequently, according to the round advancements in Mysticeti-\sysname (\Cref{sec-integrate-sysname-to-mysticeti}), where validators wait for $7\Delta$ before giving up the certificate creation in round $r{+}2$, every honest validator can create a round $r{+}2$ block that references these $2f{+}1$ round $r{+}1$ blocks from honest validators. In other words, every honest validator will create a certificate for $B_L^{r}$.
\end{proof}

\begin{lemma} \label{lem-mysticeti-round-robin}
    The round-robin schedule of leader blocks in Mysticeti-\sysname ensures that in any window of $3f{+}3$ rounds, there are three consecutive rounds with honest leader blocks.
\end{lemma}
\begin{proof}
    There are $3f{+}1$ groups of three consecutive rounds. Due to the round-robin schedule, each of the honest validators must be one of the leaders in exactly 3 of these groups. As there are $2f{+}1$ honest validators, due to the pigeonhole principle, one group must contain $\lceil \frac{3*(2f+1)}{3f+1}\rceil=3$ honest leader blocks.
\end{proof}

\begin{lemma} \label{lem-mysticeti-eventual-decide}
    In Mysticeti-\sysname, after GST, any undecided leader block eventually gets decided.
\end{lemma}
\begin{proof}
    Consider an undecided leader block in round $r$. After GST, by \Cref{lem-mysticeti-round-robin}, there will eventually be three honest leader blocks in three consecutive rounds $k$, $k{+}1$, and $k{+}2$ with $k>r$. By \Cref{lem-sufficient-votes}, each of these honest leader blocks will have $2f{+}1$ certificates and can be decided as \textit{to-commit} via the direct decision rule. We now prove that by induction, all leader blocks in rounds $<k$ get decided. For the base case, any undecided leader blocks in rounds $k{-3}$, $k{-}2$, and $k{-}1$ get decided by the \textit{to-commit} leader blocks in rounds $k$, $k{+}1$, and $k{+}2$, respectively, via the indirect decision rule. For the induction step, any undecided leader block in round $r'<k-3$ also gets decided since $k$ is higher than $r'+2$ and there are no undecided leader blocks between $r'$ and $k$ (induction hypothesis).
\end{proof}


\begin{lemma} \label{lem-accept-with-a-quorum}
    In Mysticeti-\sysname, if a block $B$ is referenced by $2f{+}1$ subsequent blocks, then every honest validator will eventually output $\dagpoa$ for $B$.
\end{lemma}
\begin{proof}
    If $B$ is referenced by $2f{+}1$ subsequent blocks, at least $f{+}1$ honest validators reference $B$. These $f{+}1$ honest blocks $\mathcal{B}_{hst}$ will eventually be received by all honest validators. According to the ImPoA-based pull protocol (\autoref{sec-sysname-pull}), $\mathcal{B}_{hst}$ forms an implicit proof-of-availability certificate for $B$ and is output via $\dagpoa$ by every honest validator.
\end{proof}

\begin{theorem}[Consensus Liveness] \label{the-consensus-liveness}
    In Mysticeti-\sysname, after GST, transactions will be ordered and finalized.
\end{theorem}
\begin{proof}
    By \Cref{lem-sufficient-votes}, there will be $2f{+}1$ certificates for each honest leader block after GST, and the honest leader block will be decided as \textit{to-commit}. By \Cref{lem-mysticeti-eventual-decide}, all leader blocks will eventually get decided. Therefore, validators can order all \textit{to-commit} leader blocks and their causal history blocks. Moreover, since each \textit{to-commit} leader block created is referenced by $2f{+}1$ subsequent blocks as parents, by \Cref{lem-accept-with-a-quorum}, every honest validator will output $\dagpoa$ for the leader block. According to block availability and causal availability ensured by \sysname, the leader block and its causal history blocks will eventually be output via $\dagdeli$. This means that all transactions in \textit{to-commit} leader blocks and their causal history blocks can be retrieved, ordered and finalized.
\end{proof}

\subsection{Safety of Mysticeti-\sysname} \label{sec-my-sys-safety}
\sysname only relates to the liveness of protocols and does not interfere with their safety, as it does not change their commit rule. Nevertheless, we then show that Mysticeti-\sysname ensures safety. The safety states that the ordered transaction sequences of any two honest validators are consistent prefixes of each other.

Note that in Mysticeti-\sysname, the following facts are inherently held by definitions: 
(i) Honest validators only create at most one block each round (i.e., honest validators do not equivocate by definition); (ii) Each block must reference as parents $2f{+}1$ blocks created by $2f{+}1$ distinct validators from the immediately preceding round; (iii) A block is valid only if its creator corresponds to a registered validator in $\mathcal{V}$ (i.e., the set of validators in the system), which is attested by the validator's signature; (iv) The block digest is derived from hashing the block and can be used to identify a block, where an identical digest implies the same block.

\begin{lemma} \label{lem-certificate-link}
    In Mysticeti-\sysname, if $2f{+}1$ round $r$ blocks from distinct validators are certificates of a block $B$, then every block in any round $r'{>}r$ must (directly or transitively) reference a certificate for $B$ formed in round $r$.

\end{lemma}
\begin{proof}
    Recall that a block is a certificate for $B$ if it references $2f{+}1$ blocks that themselves reference $B$. Consider round $r{+}1$. Every block in this round references $2f{+}1$ blocks from round $r$. By quorum intersection, any such set intersects the certificate set of $B$ in at least one honest validator. Since honest validators do not equivocate, every round $r{+}1$ block must reference a block that is a certificate for $B$.
    By induction over rounds, this property propagates to all $r'{>}r$.
\end{proof}

\begin{lemma} \label{lem-direct-skip}
    In Mysticeti-\sysname, if an honest validator directly skip a round $r$ leader block $B_L^r$, then no honest validator commits $B_L^r$.
\end{lemma}
\begin{proof}
    A direct skip occurs only if at least $2f{+}1$ blocks in round $r{+}1$ do not reference $B_L^r$. However, if $B_L^r$ is committed, then by the decision rule, at least $2f{+}1$ round $r{+}1$ blocks reference it. Any conflicting observation would violate quorum intersection, implying an honest validator equivocated in round $r{+}1$, a contradiction.
\end{proof}


\begin{lemma} \label{lem-no-skip-commit}
    In Mysticeti-\sysname, if an honest validator directly commits $B_L^r$, then no honest validator (directly or indirectly) decides to skip $B_L^r$.
\end{lemma}
\begin{proof}



    Assume for contradiction that some honest validator commits $B_L^r$ while another skips it.

    First, consider $B_L^r$ is directly skipped. A direct skip occurs only if at least $2f{+}1$ blocks in round $r{+}1$ do not reference $B_L^r$. However, if $B_L^r$ is committed, then by the decision rule, at least $2f{+}1$ round $r{+}1$ blocks reference it. Any conflicting observation would violate quorum intersection, implying an honest validator equivocated in round $r{+}1$, a contradiction.

    Then, consider $B_L^r$ is indirectly skipped. According to the indirect rule, the skip must arise through a later leader block $B_L^{r'}$ with $r' {>} r{+}2$. By the decision rule, such a block must not reference a certificate for $B_L^r$. However, since $B_L^r$ is directly committed, it has $2f{+}1$ certificates formed by distinct validators. By \Cref{lem-certificate-link}, a certificate for $B_L^r$ must be referenced by all future blocks. This contradiction completes the proof.

\end{proof}

\begin{lemma} \label{lem-one-block-certify}
    In Mysticeti-\sysname, at most one leader block can be certified for any round $r$.
\end{lemma}
\begin{proof}

    Suppose two distinct leader blocks $B_{L1}^r$ and $B_{L2}^r$ both obtain $2f{+}1$ references from round $r{+}1$. By quorum intersection, at least one honest validator must belong to both quorums, and thus would have referenced both blocks in round $r$. This contradicts the protocol rule that a validator references at most one block per proposer per round.
\end{proof}

\Cref{lem-one-block-certify} implies the following corollary.
\begin{corollary} \label{coro-commit-same-leader}
    No two honest validators commit distinct leader blocks in the same round.
\end{corollary}

\begin{lemma} \label{lem-consistent-status}
    In Mysticeti-\sysname, all honest validators decide a consistent status for each round leader block.
\end{lemma}
\begin{proof}
    Consider two honest validators $v_i$ and $v_j$. According to the decision rules, a validator commits a leader block only when the direct commit occurs (i.e., either the block or a future leader block is directly committed). Let $n$ and $m$ be the highest rounds in which $v_i$ and $v_j$ directly commit a leader block, respectively. W.l.o.g, $n\leq m$. We then prove by backward induction that for every round $x \le n$, if both validators decide the leader block $B_L^x$, then both assign the same status to $B_L^x$.

    \textbf{Base case ($x{=}n$).}
    Validator $v_i$ directly commits $B_L^n$. By \Cref{lem-no-skip-commit}, $v_j$ cannot skip it. By \Cref{coro-commit-same-leader}, if $v_j$ decides the leader block in round $n$, then $v_i$ and $v_j$ commit the same block.

    \textbf{Inductive step.}
    Assume the statement holds for all rounds in $(k, n]$. Consider round $k$.
    If an honest validator directly commits $B_L^k$, then, similar to the proof in the base case, another honest validator will commit $B_L^k$ if it decides the round $r$ leader block.
    Moreover, by \Cref{lem-direct-skip}, if an honest validator directly skips the round $k$ leader block, another honest validator will skip the round $k$ leader block if deciding it.
    The only remaining case is that both decisions are indirect and derived from a later committed leader. Let $k_i$ and $k_j$ be the rounds of the first such commits for $v_i$ and $v_j$, respectively. By the induction hypothesis, $k_i = k_j$, and both validators commit the same leader block. Since indirect decisions depend only on the causal history of that block, both validators derive the same decision for $B_L^k$.
\end{proof}

\begin{theorem}[Consensus safety] \label{the-consensus-safety}
    In Mysticeti-\sysname, for any two honest validators $v_i$ and $v_j$, let $S_i$ and $S_j$ denote their ordered transaction sequences. Then $S_i$ and $S_j$ are prefix-consistent, i.e., one is a prefix of the other.
\end{theorem}
\begin{proof}
    According to the consensus logic employed by Mysticeti-\sysname, validators schedule ordering transactions up to round $r$ only when all leader blocks with round $r' \leq r$ have been decided. W.l.o.g, for any two honest validators $v_i$ and $v_j$, we assume $n$ and $m$ are such rounds and $n\leq m$.
    By \Cref{lem-consistent-status}, both $v_i$ and $v_j$ decide a consistent status for each round leader block with round $r'' \leq n$, meaning that they decide identical \textit{to-commit} leader blocks up to round $n$. 
    As a result, both $v_i$ and $v_j$ will order these \textit{to-commit} leader blocks and their causal history block consistently.
    Therefore, we can conclude that for any two honest validators, their ordered transaction sequences are consistent prefixes of each other.
\end{proof}

\section{The Pseudocode of the AC-based Push Protocol} \label{sec-push-pseudo-code}
\Cref{fig:push-alg} provides a pseudocode for \sysname's push protocol. The key components consist of a reputation mechanism (lines~\ref{step:reputation-start}-\ref{step:reputation-end}) and an admission control (lines~\ref{step:ac-start}-\ref{step:ac-end}), both of which are detailed in \Cref{sec-sysname-push}.

\begin{figure}[t]
    \footnotesize
    \begin{boxedminipage}[t]{0.48\textwidth}
        \textbf{Variables:}
        {
            \makeatletter
            \def\@listi{\leftmargin17pt \labelwidth\z@ \labelsep5pt}    
            \makeatother
            \begin{itemize}
                \item[] $R_L - $ The score decrease each time
                \item[] $TR_i[] - $ An array of reputations (indexed by validators)
                \item[] struct block $B$ {
                              \begin{itemize}
                                  \item [] $\cdots$ \Comment{original fields}
                                  \item[] $B.weaklinks$ - used to link blocks that $v_i$ has received and accepted but not referenced as parents
                                  \item[] $B.watermark[]$ - an array of the highest round numbers of all validators' blocks received by $v_i$
                                  \item[] $B.ancestors[]$ - an array of the highest round numbers of all validators' blocks reachable from $B$
                              \end{itemize}
                          }
            \end{itemize}
        }
        \begin{algorithmic}[1]
            \Statex {\color{gray}{$\triangleright$  Call $\dagbc_i(B, r)$} to push a round $r$ block $B$}
            \Procedure{create\_new\_block}{$r$, $\mathcal{B}^{r{-}1}$} \label{step:ac-start} \Comment{$\mathcal{B}^{r{-}1}$ is a list of the latest received blocks from all validators with round $\leq r{-}1$}
            \State {Initializes a block $B$ with $B.r=r$, $B.author=i$, and other original fields}
            \State $parents \gets$ \Call{AC\_parent\_selection}{$r$, $\mathcal{B}^{r{-}1}$}
            \State $B.parents \gets$ \text{digests of} $parents$ \label{step:link-start}
            \State $B.weaklinks \gets \{B'.d | B'\in\mathcal{B}^{r{-}1}{\setminus}parents$ \text{is acceptable}\} \label{step:link-end}
            \State $watermark \gets []$
            \For{$\forall B' \in \mathcal{B}^{r{-}1}$}
            \State $watermark[B'.author] \gets B'.r$
            \EndFor
            \State $B.watermark \gets watermark$
            \State $B.ancestors \gets$ \Call{compute\_ancestors}{$parents$}
            \State signs and broadcasts $B$ using best-effort broadcast
            \State \Call{update\_score\_with\_watermarks}{$r{-}1$, $\mathcal{B}^{r{-}1}$}
            \State outputs $\dagpoa_i$ and $\dagdeli_i$ for $B$ and every acceptable block in $\mathcal{B}^{r{-}1}$, if it hasn't done so already
            \EndProcedure
            \vspace{-0.2em}

            \Procedure{AC\_parent\_selection}{$r$, $\mathcal{B}$}
            \State $\mathcal{B} \gets \{B' \in \mathcal{B} | B'.r=r{-}1 \land B' \text{is acceptable}\}$
            \State $parents \gets \text{top } 2f{+}1 \text{ blocks in } \mathcal{B} \text{ by } TR_i[B'.author]$
            \State \Return $parents$
            \EndProcedure
            \vspace{-0.2em}

            \Procedure{compute\_ancestors}{$parents$}
            \State $ancestors \gets []$
            \For{$k \in 1, \cdots, n$}
            \State $ancestors[k] \gets \max(\{B'.r | B' \in parents \land B'.author = v_k\} \cup \{B'.ancestors[k] | B' \in parents\})$
            \EndFor
            \State \Return $ancestors$ \label{step:ac-end}
            \EndProcedure
            \vspace{-0.2em}

            \Procedure{update\_score\_with\_watermarks}{$r$, $\mathcal{B}^r$} \label{step:reputation-start}
            \For{$j \in 1, \cdots, n$}
            \State $count \gets 0$
            \For{$\forall B' \in \mathcal{B}$}
            \If{$B'.watermark[j] == r{-}1$}
            \State $count \gets count + 1$
            \EndIf
            \EndFor
            \If{$count \geq 2f{+}1$}
            \State $TR_i[j] \gets TR_i[j] + 1$
            \EndIf
            \EndFor
            \EndProcedure
            \vspace{-0.2em}

            \Event{pulling or receiving $f{+}1$ pull requests (i.e., blames) for a missing block created by $v_j$}
            \State $TR_i[j] \gets TR_i[j] - R_L$    \label{step:reputation-end}
            \EndEvent
            \vspace{-0.2em}


        \end{algorithmic}
    \end{boxedminipage}
    \caption{\sysname's AC-based optimistic push protocol for validator $v_i$.}
    \label{fig:push-alg}
\end{figure}
\section{\sysname Implementation} \label{sec:implementation}
We implement \sysname in Rust within Mysticeti~\cite{babel2025mysticeti} by forking the Mysticeti codebase~\cite{mysticeti-code}. It leverages \texttt{tokio}~\cite{tokio} for asynchronous networking, utilizing raw TCP sockets for communication, implementing reliable point-to-point channels, necessary to correctly implement the distributed system abstractions without relying on any RPC frameworks. For cryptographic operations, it relies on \texttt{ed25519-consensus}~\cite{ed25519-consensus} for asymmetric cryptography and \texttt{blake2}~\cite{rustcrypto-hashes} for cryptographic hashing. To ensure data persistence and crash recovery, it employs a Write-Ahead Log (WAL) optimizing I/O operations through vectored writes~\cite{writev}, efficient memory-mapped files, and uses \texttt{minibytes}~\cite{minibytes} to minimize copies and serialization.

By default, this Mysticeti implementation uses a traditional optimistic push followed by a random pull protocol (described in \Cref{sec-problem-def}). We modify its block synchronizer module to use \sysname instead. Implementing our mechanism requires adding less than 200 LOC, and does not require any extra cryptographic tools.

In addition to regular unit tests, we inherited and utilized two supplementary testing utilities from the Mysticeti codebase. First, a simulation layer replicates the functionality of the \texttt{tokio} runtime and TCP networking. This simulated network accurately simulates real-world WAN latencies, while the \texttt{tokio} runtime simulator employs a discrete event simulation approach to mimic the passage of time. Second, a command-line utility (called \emph{orchestrator}) which deploys real-world clusters of \sysname on machines distributed across the globe.

We are open-sourcing our \sysname implementation, along with its simulator and orchestration tools, to ensure reproducibility of our results\footnote{\codelink}.
\section{Experimental Setup} \label{sec:experimental-setup}
This section describes the experimental setup used for evaluating \sysname in \Cref{sec-experiment}.

We deploy all systems on AWS, using \texttt{m5d.8xlarge} instances across $13$ different AWS regions: Northern Virginia (us-east-1), Oregon (us-west-2), Canada (ca-central-1), Frankfurt (eu-central-1), Ireland (eu-west-1), London (eu-west-2), Paris (eu-west-3), Stockholm (eu-north-1), Mumbai (ap-south-1), Singapore (ap-southeast-1), Sydney (ap-southeast-2), Tokyo (ap-northeast-1), and Seoul (ap-northeast-2). Validators are distributed across those regions as equally as possible.
Each machine provides $10$\,Gbps of bandwidth, $32$ virtual CPUs (16 physical cores) on a $3.1$\,GHz Intel Xeon Skylake 8175M, $128$\,GB memory, and runs Linux Ubuntu server $24.04$.
We select these machines because they provide decent performance, are in the price range of ``commodity servers'', and match the minimal specifications of modern quorum-based blockchains~\cite{sui-min-specs}.

The \emph{latency} refers to the time elapsed from the moment a client submits a transaction to when it is committed by the validators, and the \emph{throughput} refers to the number of transactions committed per second.
We instantiate several geo-distributed benchmark clients within each validator submitting transactions in an open loop model, at a fixed rate. We experimentally increase the load of transactions sent to the systems, and record the throughput and latency of commits. As a result, all plots illustrate the steady-state latency of all systems under various loads. Transactions in the benchmarks are arbitrary and contain $512$ bytes. We configure both systems with $2$ leaders per round, and a leader timeout of 1 second.
\section{Pull Induction Attacks in \synccert Synchronizer Protocol}
In this section, we demonstrate how pull induction attacks can be extended to the current implementation of \synccert synchronizer protocol~\cite{narwhal-code} that has been widely adopted by certified DAG protocols, such as Narwhal~\cite{danezis2022narwhal}, Bullshark~\cite{spiegelman2022bullshark}, shoal~\cite{spiegelman2025shoal}, and Shoal++~\cite{arun2024shoal++}.

\subsection{Review CBC-based Push Mechanism}
The \synccert protocol adopts a CBC-based push mechanism, where each validator $v_i$ pushes its block $B$ to other validators each round $r$ through a three-step scheme, as shown in \Cref{fig:cbc-happy}. 

\para{Step 1: Broadcast block.}
For each round $r$, the validator $v_i$ creates a new block $B_i^r$, which specifies the round number $r$ and the creator $i$, uses round $r-1$ blocks as $parents$, and includes a list of new transactions $payload$. $v_i$ then signs $B_i^r$ and broadcasts it to other validators.

\para{Step 2: Collect acknowledgments.}
Upon receiving $B_i^r$, each other validator $v_j$ checks its validity. A block is valid if it (i) contains a valid signature from its creator, (ii) is at the local round $r$ of the validator checking it, (iii) contains at least $2f+1$ parent blocks of round $r-1$, and (iv) is the first one received from the creator for round $r$. $v_j$ then checks whether the availability of a valid block $B_i^r$ can be guaranteed. Specifically, if $v_j$ has output $\dagpoa$ for each parent block of $B_i^r$, then $B_i^r$ is guaranteed to be available, and $v_j$ will acknowledge it by signing its block digest, round number, and the creator's identity. Otherwise, $v_j$ has to invoke the pull procedure to synchronize missing ancestors. Eventually, $v_j$ sends its acknowledgment of $B_i^r$ back to $v_i$ for collection.

\para{Step 3: Broadcast certificate.}
Once $v_i$ gets $2f+1$ distinct acknowledgments for $B_i^r$, it combines them into a certificate $C(B_i^r.d)$. $v_i$ then broadcasts the certificate to other validators. Each other validator $v_j$ can independently verify the availability of $B_i^r$ with the certificate, and will output $\dagpoa$ for $B_i^r$.

Once a validator $v_i$ moves to a new round $r'$, it will use blocks as parents from round $r'-1$ that it has output $\dagpoa$. Validators repeat the above steps to continuously push new blocks into the network.

\begin{figure}[t]
    \centering

    \begin{subfigure}{\linewidth}
        \centering
        \includegraphics[width=\linewidth]{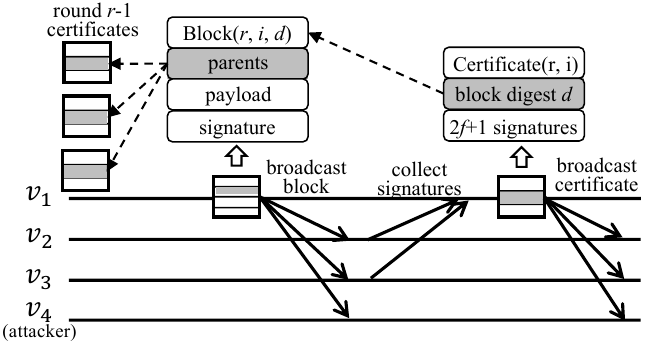}
        \caption{Honest validator $v_1$ pushes blocks.}
        \label{fig:cbc-happy}
    \end{subfigure}

    \vspace{-0.0em}

    \begin{subfigure}{\linewidth}
        \centering
        \includegraphics[width=\linewidth]{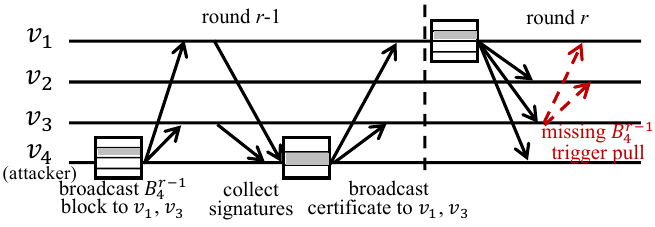}
        \caption{Malicious validator $v_4$ launches the pull induction attack.}
        \label{fig:cbc-attack}
    \end{subfigure}

    \caption{Illustration of \synccert synchronizer protocol under happy cases and attack cases. (a) The honest validator $v_1$ shares its block to all other validators. (b) The malicious validator $v_4$ selectively shares its round $r-1$ block $B_4^{r-1}$ to $v_1$ and $v_3$ but not $v_2$. In round $r$, $v_1$ references $B_4^{r-1}$ as parents and sends its round $r$ block $B_1^r$ to $v_2$. $V_2$ has to invoke the pull procedure to get $B_4^{r-1}$ before it can output $\dagpoa$ for $B_1^r$. As a result, $v_2$ is stuck in round $r$, and the latency of round $r$ is increased by one pull round-trip.}
    \label{fig:cbc-push}
\end{figure}

\subsection{Pull Induction Attacks}
The current implementation of the \synccert protocol, however, is vulnerable to the pull induction attack. \Cref{fig:cbc-attack} illustrates how a malicious validator $v_4$ can launch the pull induction attack to increase the push latency. 

\para{Attack example.}
In round $r-1$, $v_4$ only disseminates its round $r-1$ block $B_4^{r-1}$ to $v_1$ and $v_3$ but not $v_2$. $v_4$ collects $2f+1$ acknowledgments (including itself) and is able to construct a certificate for $B_4^{r-1}$. As a result, $v_1$ and $v_3$ will accept $v_4$'s certificate for $B_4^{r-1}$ and reference $B_4^{r-1}$ as parents of their round $r$ blocks. Once moving to round $r$, $v_1$ broadcasts its block $B_1^r$ (referencing $B_4^{r-1}$) to $v_2$. However, since $v_2$ has not received $B_4^{r-1}$, according to the CBC-based push mechanism (Step 2), it will not acknowledge $B_1^r$ immediately. Instead, $v_2$ triggers the pull procedure to synchronize $B_4^{r-1}$. This will introduce an extra round-trip (i.e., $2\delta$ due to its deterministic pull mechanism) to the round $r$. As a result, the push latency of round $r$ will be increased to $5\delta$.

\para{Discussion.}
The \synccert synchronizer protocol might prevent the above attacks by introducing extra overhead and complexity. For instance, instead of merely including the block digests as parents, validators can include the block digest along with its $2f+1$ acknowledgments/signatures as parents. This allows a validator to verify the availability of each block it receives without triggering the pull procedure on the push path, as these $2f+1$ acknowledgments themselves attest to the availability of the block. However, including $2f+1$ signatures into blocks will increase the protocol communication complexity, i.e., from $O(n^3)$ to $O(n^4)$, making the protocol less scalable and practical. On the other hand, one might use the threshold signature scheme to avoid the communication overhead introduced by the above approach. However, this requires a costly threshold-signature setup with either a trusted assumption or distributed key generation and resharing on every membership change. In contrast, \sysname is free from the explicit certificate and offers better push latency while resisting to pull induction attacks.
\section{Lean Formalisation} \label{sec-lean-formalisation}
We provide a complete machine-checked formalisation of the non-probabilistic results in this paper, developed in Lean~\cite{Moura021}. The proof artefacts are available with the accompanying material (please see `Additional materials' we have uploaded). This appendix records what is formalised, what is not, and the places where the formalisation makes a paper-implicit choice explicit.

\subsection{What is Formalised} \label{sec-lean-prove}

The Lean formalisation features the following components:

\begin{itemize}
    \item \textbf{\autoref{sec-block-sync-def} Definitions.} \Cref{def-dag-sync} and its four projections — Round-Progression (1.1), Round-Termination (1.2), Block-availability (1.3), and Causal-availability (1.4) — together with the network model, the honest and Byzantine partition, the synchronizer interface ($\dagbc$, $\dagpoa$, and $\dagdeli$), and the causal-history relation.
    \item \textbf{\autoref{sec-sysname} Protocol semantics.} The \sysname protocol of \autoref{sec-sysname} is modeled in full: block extensions of \autoref{sec-sysname-overview}, the reputation mechanism and admission control of \autoref{sec-sysname-push}, the ImPoA-based hybrid pull of \autoref{sec-sysname-pull}, and the \autoref{sec-sysname-consensus} availability and certificate patterns. The protocol is given as both a relational specification and an executable step-function that refines it.
    \item \textbf{\autoref{sec-analysis} Main theorems}. \Cref{lem-round-liveness} (round-entry within $4\Delta$), \Cref{lem-block-availability}–\Cref{thm-round-termination} (\sysname satisfies block-availability, causal-availability, round-progression, and round-termination), and the corollary that \sysname is a block synchronizer in the sense of \Cref{def-dag-sync}.
    \item \textbf{\Cref{sec-performance-analysis} deterministic bounds.} \Cref{asm-latency-triangle} (the latency triangle), \Cref{lem-no-blame-honest}, \Cref{lem-optimal-round-latency-condition}, and the deterministic disjunct of \Cref{lem-bounded-latency-each-round} (``after-GST round latency $2\Delta$-or-blame'').
    \item \textbf{\Cref{sec-sec-analysis} Mysticeti-\sysname.} The \Cref{sec-mysticeti-sysname-review} consensus rules (direct/indirect decision, skip and certificate patterns, the round-robin leader schedule), the \Cref{sec-my-sys-liveness} liveness chain (\Cref{lem-leader-proposal}, \Cref{lem-sufficient-votes}, \Cref{lem-mysticeti-round-robin}, \Cref{lem-accept-with-a-quorum} and \Cref{the-consensus-liveness}), and the \Cref{sec-my-sys-safety} safety chain (Lemmas~\ref{lem-certificate-link}–\ref{lem-consistent-status} and \Cref{the-consensus-safety}).
\end{itemize}

\subsection{What is Not Formalised} \label{sec-lean-not-prove}

The formalisation is restricted to the deterministic content of the paper. The following items, all of which involve probabilistic or expected-value reasoning, have been omitted from the formalisation (as they would require building, from scratch, an extensive library for reasoning about probabilistic network semantics, which is out of scope for this work):
\begin{itemize}
    \item The expected-latency disjunct of \Cref{lem-bounded-delay} in \Cref{sec-performance-analysis}.
    \item \Cref{lem-bounded-number{-}of-delayed-rounds} and \Cref{thm-optimal-round-latency} from \Cref{sec-performance-analysis}, providing expected-latency upper bounds.
\end{itemize}

In addition, \Cref{lem-mysticeti-eventual-decide} from \Cref{sec-sec-analysis}  is mechanised in its existential form: after-GST, at every starting round, there \emph{exists} a future round at which the leader's block is direct-committed (rather than in its universal indirect-rule form). The existential form is what the proof of \Cref{the-consensus-liveness} actually requires (\Cref{the-consensus-liveness} derives every honest validator's eventual acceptance from \Cref{sec-analysis} in-pool delivery and \Cref{sec-sysname-push} accept-action liveness, without invoking the indirect-rule chain). The recursive descent of the \Cref{sec-consensus-logic-mysticeti} indirect-decision rule is therefore not mechanised, as it's not required by any downstream proofs.

\subsection{Modelling Details}
The formalisation pins down two minor choices the paper leaves implicit; neither weakens nor strengthens any paper claim.
\begin{itemize}
    \item \textbf{Block representation.} A block in the paper carries ($r$, $d$, $author$, $parents$, $payload$, $signature$, $\cdots$). The formalisation retains the first five fields and routes Byzantine behaviour through the honest / Byzantine partition rather than through signature attribution; no theorem we mechanise invokes signature semantics.
    \item \textbf{Availability pattern.} \autoref{sec-sysname-consensus} defines a block as referenced by another when the latter strongly links or weak links to it. The formalisation counts strong-link references only, which yields a predicate that lower-bounds the paper's; every block forming the formalisation's pattern forms the paper's pattern.
\end{itemize}

\fi

\end{document}